\documentclass[manuscript,screen,nonacm]{acmart}

\usepackage{booktabs} 
\usepackage[utf8]{inputenc}
\usepackage{amsmath,amsthm,amssymb}
\usepackage[ruled,vlined]{algorithm2e}
\usepackage{subcaption}
\usepackage{balance}
\usepackage{graphicx}
\usepackage{url}
\usepackage{color}
\usepackage{tikz}

\DeclareMathOperator*{\argmax}{arg\,max}
\DeclareMathOperator*{\argmin}{arg\,min}
\DeclareMathOperator{\dist}{d}
\DeclareMathOperator{\rank}{rank}
\DeclareMathOperator{\MST}{MST}
\DeclareMathOperator{\TSP}{TSP}

\renewcommand{\epsilon}{\varepsilon}
\DeclareMathOperator{\diversity}{div}
\newcommand{\BO}[1]{O\left( #1 \right)}
\newcommand{\BT}[1]{\Theta\left( #1 \right)}

\newcommand{\algo}[1]{{\scshape #1}}

\newtheorem{definition}{Definition}

\newtheorem{theorem}{Theorem}
\newtheorem{lemma}{Lemma}
\newtheorem*{lemma*}{Lemma}

\newtheorem{fact}{Fact}

\newcommand{\Let}[2]{#1 $\leftarrow$ #2}

\newif\ifshort
\shortfalse

\begin{document}

\title{A General Coreset-Based Approach to Diversity Maximization under Matroid Constraints}

\author{Matteo~Ceccarello}
\email{mceccarello@unibz.it}
\orcid{0000-0003-2783-0218}
\affiliation{%
  \institution{Free University of Bozen}
  \city{Bolzano}
  \state{Italy}
}

\author{Andrea Pietracaprina}
\email{andrea.pietracaprina@unipd.it}
\author{Geppino Pucci}
\email{geppino.pucci@unipd.it}
\affiliation{%
  \institution{University of Padova}
  \department{Department of Information Engineering}
  \city{Padova}
  \state{Italy}
}

\renewcommand{\shortauthors}{Ceccarello, Pietracaprina, Pucci}

\begin{abstract}
Diversity maximization is a fundamental problem in web search and data
mining.  For a given dataset $S$ of $n$ elements, the problem requires
to determine a subset of $S$ containing $k\ll n$ "representatives"
which minimize some diversity function expressed in terms of pairwise
distances, where distance models dissimilarity. An important variant
of the problem prescribes that the solution satisfy an additional
orthogonal requirement, which can be specified as a matroid constraint
(i.e., a feasible solution must be an independent set of size $k$ of a
given matroid). While unconstrained diversity maximization admits
efficient coreset-based strategies for several diversity functions,
known approaches dealing with the additional matroid constraint apply
only to one diversity function (sum of distances), and are based on an
expensive, inherently sequential, local search over the entire input
dataset. We devise the first coreset-based algorithms for diversity
maximization under matroid constraints for various diversity functions,
together with efficient sequential, MapReduce and Streaming
implementations. Technically, our algorithms rely on the construction
of a small coreset, that is, a subset of $S$ containing a feasible
solution which is no more than a factor $1-\epsilon$ away from the
optimal solution for $S$. While our algorithms are fully general, for
the partition and transversal matroids, if $\epsilon$ is a constant in
$(0,1)$ and $S$ has bounded doubling dimension, the coreset size is
independent of $n$ and it is small enough to afford the execution of a
slow sequential algorithm to extract a final, accurate, solution in
reasonable time.  Extensive experiments show that our algorithms are
accurate, fast and scalable, and therefore they are capable of
dealing with the large input instances typical of the big data
scenario.

\end{abstract}



\maketitle

\section{Introduction}

In many application domains, data analysis often requires the
extraction of a succinct and significant summary of a large dataset
which may take the form of a small subset of elements as diverse as
possible from one another.  The summary can be either presented to the
user or employed as input for further
processing~\cite{AbbassiMT13,MasinB08,wu2013,YangMNFCH15}. More
specifically, given a dataset $S$ of points in a metric space and a
constant $k$, \emph{diversity maximization} requires to determine a
subset of $k$ points of $S$ maximizing some diversity
objective function defined in terms of the distances between the
points. 

There are several ways of characterizing the diversity function.  In
general, the diversity of a set of $k$ points can be captured by a
specific graph-theoretic measure defined on the points, which are seen
as the nodes of a $k$-clique where each edge is weighted with the
distance between its endpoints~\cite{ChandraH01}. The diversity
functions considered in this paper are defined in
Table~\ref{tab:diversity-notions}. The maximization problems 
under these functions are all known to be NP-hard \cite{ChandraH01}
and the development of efficient approximation algorithms
has attracted a lot of interest in the recent literature
(see \cite{CeccarelloPPU17} and references therein).

An important variant of diversity maximization requires that the $k$
points to be returned satisfy some additional orthogonal constraint
such as, for example, covering a variety of pre-specified categories
attached to the data.  This variant has been recently investigated
under the name \emph{Diversity Maximization under Matroid Constraint}
(\emph{DMMC}) where the subset of $k$ points to be returned is
required to be an independent set of a given matroid
\cite{BorodinLY12,AbbassiMT13,CevallosEZ17}. (A more formal definition
of the problem is given in Section~\ref{sec:preliminaries}.)  As a
concrete example, suppose that $S$ is a set of Wikipedia pages, each
associated with one or more topics. A solution of the DMMC problem
identifies a subset of pages that are most diverse in terms of some
pre-specified distance metric (e.g., cosine distance
\cite{LeskovecRU14}) but also ``well spread'' among the topics. This
latter property can be suitably controlled by imposing a partition or
transversal matroid constraint, depending on whether topics overlap or
not.

In this paper, we contribute to this line of work and present novel
algorithms for diversity maximization under matroid constraint for the
diversity functions in Table~\ref{tab:diversity-notions}, both for the
traditional sequential setting and for big-data oriented computation
frameworks such as MapReduce \cite{Dean2004} and Streaming \cite{HenzingerRR98}.

\subsection{Related work}

Unconstrained diversity maximization has been studied for over two
decades within the realm of facility location (see \cite{ChandraH01}
for an account of early results). In recent years, several works have
devised efficient algorithms for various diversity maximization
problems in different computational frameworks. Specifically, in
\cite{CevallosEM18} PTAS's are devised in the sequential setting
for metric spaces of constant doubling dimension (a notion that will
be formalized in Subsection~\ref{sec:ddimension}). In
\cite{IndykMMM14,AghamolaeiFZ15,CeccarelloPPU17,EpastoMZ19}, MapReduce
and/or Streaming algorithms are proposed for several diversity
measures, which are based on confining the expensive computations
required by standard sequential algorithms to small subsets of the
input (\emph{coresets}), cleverly extracted so to contain high-quality
global solutions.  For general metric spaces, the algorithms in
\cite{IndykMMM14,AghamolaeiFZ15,EpastoMZ19}, require sublinear working
memory per processor at the expense of a constant worsening in the
approximation ratio with respect to the best ratio attained by
sequential algorithms, while, for metric spaces of constant doubling
dimension, the algorithms in \cite{CeccarelloPPU17} retain sublinear
space but feature approximation ratios that can be made arbitrarily
close to the best sequential ones.  Finally, unconstrained diversity
maximization is also studied in \cite{BorassiELVZ19} in the
sliding-window framework.

The literature on diversity maximization under matroid constraints is
much thinner.  Unlike the unrestricted case, existing approaches only
target the sum-DMMC variant (see Table~\ref{tab:diversity-notions}) in
the sequential setting, and are based on expensive local search
strategies over the entire input $S$. Specifically, both
\cite{BorodinLY12} and \cite{AbbassiMT13} present (1/2)-approximation
algorithms for sum-DMMC whose running times are at least quadratic in
the input size, hence impractical for large inputs.  In fact, these
algorithms guarantee polynomial time only at the expense of a slightly
worse approximation ratio $1/2 -\gamma$, for any fixed $\gamma >0$. In
the preliminary conference version of our work \cite{CeccarelloPP18},
sequential, MapReduce and Streaming algorithms for the sum-DMMC variant have
been presented, which achieve the same approximation quality 
as \cite{BorodinLY12,AbbassiMT13}. 
For metric spaces of constant doubling dimension,
the sequential algorithm is much faster 
than its competitors, while the MapReduce/Streaming algorithms
enable the processing of massive datasets 
in 2 rounds/1 pass. The
improvements of this present work over \cite{CeccarelloPP18} are
discussed in more detail at the end of the next subsection.

Recently, an extension of the above local search approaches for
sum-DMMC to non-metric spaces with negative-type distances has been
proposed in \cite{CevallosEZ17}. To the best of our knowledge, no
polynomial-time algorithms featuring nontrivial approximation
guarantees are known for the other DMMC variants listed in
Table~\ref{tab:diversity-notions}.

Finally, it has been proved that, under the widely accepted \emph{planted
  clique hypothesis}, it is hard to attain an approximation ratio
larger than 1/2 for sum-DMMC \cite{BorodinLY12,BhaskaraGMS16} (in fact, this
inapproximability result holds also for the unconstrained case). 

\subsection{Our contribution}

In this paper, we present coreset-based strategies which can be
employed to provide good approximations to all DMMC variants listed in
Table~\ref{tab:diversity-notions}, and which are amenable to efficient
implementations in the sequential, MapReduce and Streaming settings.
For all variants, the coreset constructions revolve around the same
key idea of clustering the input dataset into subsets of close-by
points, and then selecting suitable representatives from each subset,
depending on the matroid type. The essence of this idea was pioneered
in previous work on unconstrained diversity maximization
\cite{IndykMMM14,AghamolaeiFZ15,CeccarelloPPU17}.  More specifically,
for any DMMC instance and every chosen $\epsilon \in (0,1)$, our
approach builds a $(1-\epsilon)$-coreset, that is, a coreset which
contains a feasible solution to the instance whose diversity is within
a factor $(1-\epsilon)$ from the optimal diversity. The coreset size
is analyzed in terms of the matroid type, the size $k$ of the
solution, and the doubling dimension $D$ of the input dataset $S$ (see
Subsection~\ref{sec:ddimension} for a formal definition of doubling
dimension).  For constant $D$, our coreset constructions can be
implemented using work linear in $n$ and polynomial in $k$ and
$1/\epsilon$ in all three computational settings. It is important to
remark that while $k$ and $\epsilon$ are given in input together with
the dataset $S$, the value $D$ (hard to estimate in practice) 
is used in only in the analysis and needs not be
explicitly provided to the algorithms. The constructions can be
accomplished in MapReduce in one round with sublinear memory, and in
Streaming in a single pass with working memory proportional to the
(small) coreset size.  While our coreset constructions are fully
general, for the important cases of the partition and transversal
matroids and of datasets with constant $D$, the resulting coreset size
becomes independent of $n=|S|$. In the MapReduce setting, the coreset
size also depends on the degree of available parallelism $\ell$, which
can in turn be a function of $n$. However, a coreset size independent
of $n$ can always be achieved regardless of the value of $\ell$ in an
extra round, by performing a second (sequential) coreset construction
on the first corset.

In all three settings, once a $(1-\epsilon)$-coreset $T$ is computed
for $S$, the final solution can be obtained by running on $T$ a
sequential algorithm for the DMMC variant under consideration. For
sum-DMMC, running the algorithm of \cite{AbbassiMT13} on $T$ yields a
$(1/2-\BO{\epsilon})$-approximation, with total work (including the
coreset construction phase) which is linear in $n$ and, for constant
$D$, polynomial in $k$ and $1/\epsilon$. For the other DMMC variants
of Table~\ref{tab:diversity-notions}, since no nontrivial
polynomial-time approximations are known, we resort to running an
(exact) exhaustive search for the best solution on $T$. This approach
yields a $(1-\epsilon)$-approximation, with total work which is linear
in $n$ and, for constant $D$, polynomial in $1/\epsilon$ but
exponential in $k$. For small values of $k$, a range of definite
interest for real applications, ours are the first feasible algorithms
providing provably accurate solutions for these DMMC
variants. Finally, we remark that our MapReduce and
Streaming algorithms provide the first practically viable approaches to
the solutions of all DMMC variants in the big data scenario.

Our theoretical results are complemented with extensive experiments on
real-world datasets. For concreteness, we focus on the sum-DMMC
variant, the only variant for which there is a known sequential
competitor \cite{AbbassiMT13}. In the sequential and Streaming
settings, the experiments provide clear evidence that the accuracy scales
with the coreset size, in the sense that larger coreset sizes afford
solutions of higher quality. Moreover, for a given target accuracy,
both our sequential and Streaming algorithms run up to two orders of
magnitude faster than the pure local-search of \cite{AbbassiMT13}.  In
the MapReduce setting, the experiments show scalability of
performance with respect to the available parallelism.
In essence, the experiments confirm that, on sufficiently large
instances, the time for the extraction of the final solution, now confined to a small coreset  rather than the entire input, becomes
negligible with respect to overall running time, while
the running time is dominated by the highly-scalable, linear-work coreset
construction.

\vspace*{0.3cm}
\noindent
{\bf Novelty with respect to conference version.}  The novel results
of this work over those present in the preliminary conference version
\cite{CeccarelloPP18} are the following.  (a) The results have been
generalized to several diversity functions, through the introduction
of the notion of \emph{average farness} and the derivation of its relation to
the diameter of the dataset (Lemma~\ref{lem:diameter}). As a
consequence of this generalization, the current paper presents the
first feasible algorithms in the literature providing provably
accurate solutions for the DMMC variants associated to these
functions. (b) The streaming implementation of the coreset
construction is novel and simpler. More importantly, unlike to the one
presented in \cite{CeccarelloPP18}, it is oblivious to the doubling
dimension $D$ of the dataset. (c) An experimental analysis of the
streaming algorithm and an extensive comparison between all of our algorithms
have been added.

\subsection{Organization of the paper}
The rest of the paper is structured as follows. A formal definition of
the problem and some key concepts and notations are given in
Section~\ref{sec:preliminaries}. The coresets constructions are
described in Section~\ref{sec:coresets}, while their implementations
in the various settings and the resulting DMMC algorithms are
presented in Section~\ref{sec:implementation}. The experimental
results are reported in Section~\ref{sec:experiments}.
Section ~\ref{sec:conclusions} closes the paper with some final
remarks and open problems.

\section{Preliminaries}
\label{sec:preliminaries}

\subsection{Matroids}
Let $S=\{s_1, s_2, \dots, s_n\}$ be a set from a metric space with
distance function $\dist(\cdot,\cdot)$. Recall that $\dist$ is
nonnegative, symmetric, equal to 0 only on pairs of identical
elements, and obeys the triangle inequality. A
\emph{matroid}~\cite{Oxley06} based on $S$ is a pair $\mathcal{M} =
(S, \mathcal{I}(S))$, where $\mathcal{I}(S)$ is a family of subsets of
$S$, called \emph{independent sets}, satisfying the following
properties: (i) the empty set is independent; (ii) every subset of an
independent set is independent (\emph{hereditary property}); and (iii)
if $A \in \mathcal{I}(S)$ and $B \in \mathcal{I}(S)$, and $|A| > |B|$,
then there exist $x \in A\setminus B$ such that $B \cup \{x\} \in
\mathcal{I}(S)$ (\emph{augmentation property}). An independent set is
maximal if it is not properly contained in another independent set. A
basic property of a matroid $\mathcal{M}$ 
is that all of its maximal independent sets
have the same size, which is called the \emph{rank} of the matroid and
is denoted by $\operatorname{rank}(\mathcal{M})$.  In this paper we
concentrate on two well-known types of matroid, namely,
\emph{partition matroids} and \emph{transversal matroids}, which are
defined as follows.

\begin{definition}[Partition Matroid]
Consider a partition of $S$ into $h$ disjoint 
subsets $A_1, A_2, \dots, A_h$, and
let $k_i\leq |A_i|$ be a nonnegative integer, for $1\leq i\leq
h$.  Define $\mathcal{I}(S)$ as the family of subsets $X \subseteq S$ with $|X
\cap A_i| \le k_i$, for $1\leq i\leq h$.  Then, 
$\mathcal{M} =(S, \mathcal{I}(S))$ is a \emph{partition matroid} based on $S$.
\end{definition}

\begin{definition}[Transversal Matroid]
Consider a covering family $\mathcal{A}=\{A_1, \dots, A_h\}$ of (possibly
non-disjoint) subsets of $S$, that is, $S=\bigcup_{i=1}^{h}A_i$, and
consider the bipartite graph $(S, \mathcal{A}; E)$ where $E$ consists
of all edges $\{s_i, A_j\}$ with $s_i \in A_j$, for $1\leq i\leq n$
and $1\leq j \leq$.  Define $\mathcal{I}(S)$ as the family of subsets $X
\subseteq S$  corresponding to the left endpoints of some matching
in the above graph.  Then, $\mathcal{M} = (S, \mathcal{I}(S))$ is a
\emph{transversal matroid} based on $S$.
\end{definition}

In the following, the sets $A_1, A_2, \dots, A_h$ in the definition of
partition and transversal matroids will be referred to as
\emph{categories}. We make the reasonable assumption that for the
transversal matroids each element of the input belongs to a constant
number of categories. Also, without loss of generality, we assume that
each individual element of $S$ makes a singleton independent set. In
fact, elements for which this is not the case can be eliminated since,
by the hereditary property of matroids, they cannot belong to larger
independent sets.

\subsection{Problem definition}
Let $\diversity: 2^S \rightarrow \mathbb{R}$ be a \emph{diversity
  function} that maps any subset $X \subset S$ to some nonnegative real
number.  For a specific diversity function $\diversity$, a matroid
$\mathcal{M} = (S, \mathcal{I}(S))$, and a positive integer $k \leq
\operatorname{rank}(\mathcal{M})$, the goal of the \emph{Diversity
  Maximization problem under Matroid Constraint} (DMMC problem, for
brevity) is to find an independent set $X \in I(S)$ of size $k$ that
maximizes $\diversity(X)$.  We denote the optimal value of the
objective function as
\[
  \diversity_{k,\mathcal{M}}(S) = \max_{X \in \mathcal{I}(S), |X| = k} \diversity
(X)
\]
In this paper, we will focus on several instantiations of the
DMMC problem presented in Table~\ref{tab:diversity-notions}, characterized
by different diversity functions amply studied in the
previous literature\footnote{%
Observe that function
$\diversity(X) = \min_{u \neq v\in X} d(u, v)$,
also well studied in the literature for the unconstrained variant of the 
problem
is missing in the table since we were not able to 
obtain meaningful results for it. We will discuss this
issue in the conclusions.
} \cite{ChandraH01,IndykMMM14,AghamolaeiFZ15}.
Throughout the paper, 
the generic term ``DMMC problem''
will be used whenever a statement 
applies to all instantiations of Table~\ref{tab:diversity-notions}.
\newcommand{\definitionTableVerticalSpacing}{\rule{0pt}{10pt}}
\begin{table}[t]
  \centering
  \begin{tabular}{l|l}
    \toprule
    {\bf Problem}  & {\bf Diversity function} $\diversity(X)$
    \\\hline
    \definitionTableVerticalSpacing%
    sum-DMMC
      & $\sum_{u, v\in X} \dist(u, v)$ 
    \\
    \definitionTableVerticalSpacing%
    star-DMMC 
      & $\min_{c\in X}\sum_{u\in X\setminus\{c\}} \dist(c, u)$ 
    \\
    \definitionTableVerticalSpacing%
    tree-DMMC 
      & $w(\MST(X))$ 
    \\
    \definitionTableVerticalSpacing%
    cycle-DMMC
      & $w(\TSP(X))$
    \\
    \definitionTableVerticalSpacing%
    bipartition-DMMC
      & $\min_{Q\subset X, |Q|=\lfloor|X|/2\rfloor}\sum_{u\in Q, v\in X \setminus Q} \dist(u, v)$
    \\
    \bottomrule
  \end{tabular}
  \caption{%
    Instantiations of the DMMC problem considered in this paper, with
related diversity measures.  $w(\MST(X))$ (resp.,
    $w(\TSP(X))$) denotes the minimum weight of a spanning tree (resp.,
    Hamiltonian cycle) of the complete graph whose nodes are the points of
    $X$ and whose edge weights are the pairwise distances among the
    points.  
  }
\label{tab:diversity-notions}
\end{table}

Returning to the example mentioned in the introduction, concerning a
set $S$ of Wikipedia pages, the covering of various topics, viewed as
categories, can be enforced by a partition matroid constraint when
each page is labeled by a single topic, or by a transversal matroid
constraint when pages may refer to multiple topics.
  
The algorithms presented in this paper use clustering as a
subroutine. For a given positive integer $\tau$, a $\tau$-\emph{clustering}
of $S$ is a pair $(\mathcal{C},Z)$, where $\mathcal{C}=\{C_1, \dots,
C_{\tau}\}$ is a partition of $S$, and $Z = \{z_1, \dots, z_{\tau}\} \subset S$
is such that $z_i\in C_i$, for $1\leq i\leq \tau$. Each $z_i$ is said to
be the \emph{center} of its respective \emph{cluster} $C_i$. We define
the \emph{radius} of the clustering as
\[
r(\mathcal{C},Z) = \max_{1\leq i\leq \tau}\max_{s \in C_i} \dist(s, z_i).
\]
The problem of finding a $\tau$-clustering of minimum radius is
NP-hard, as it is also NP-hard to achieve an approximation factor of
$2-\epsilon$ in general metric spaces, for any $\epsilon >
0$~\cite{Gonzalez85}. In the paper, we will make use of the well-known
sequential 2-approximation clustering algorithm of \cite{Gonzalez85}
(known as \algo{gmm} in the literature) as a key tool in both our
sequential and MapReduce algorithms for the DMMC problem.  Instead, in
the streaming setting we rely on a strategy reminiscent of the
seminal streaming clustering algorithm of \cite{CharikarCFM04}. In
fact, both \cite{Gonzalez85} and \cite{CharikarCFM04} limit themselves
to identifying a suitable set $Z = \{z_1,
\dots, z_{\tau}\}$ of centers which implicitly induce a clustering
$(\mathcal{C}=\{C_1, \dots, C_{\tau}\},Z)$ 
with the desired approximation quality, where 
each $C_i$ is the set of elements which are closer to $z_i$ than any
other center.

\subsection{Doubling dimension}
\label{sec:ddimension}
Our algorithms will be analyzed in terms of the dimensionality
of the dataset $S$, as captured by the well-established notion of
doubling dimension. Formally, for a given 
point $x \in S$, let the \emph{ball of radius $r$ centered at
  $x$} be the subset of points of $S$ at distance 
at most $r$ from $x$. The
\emph{doubling dimension} of $S$ is defined as the smallest value $D$ such
that any ball of radius $r$ centered at an element $x$ is covered by
at most $2^D$ suitably centered balls of radius $r/2$. Observe that the
doubling dimension of a dataset $S$ of size $n$
is upper bounded by $\log_2 n$. The algorithms that will be presented in this
paper adapt automatically to the doubling dimension $D$ of the input dataset and 
attain their best performance when $D$ is small, possibly constant. 
This is the case, for instance, of datasets $S$ whose points belong to
low-dimensional Euclidean spaces, or represent nodes of mildly-expanding
network topologies under shortest-path distances. 
The characterization of datasets (or metric spaces)
through their doubling dimension has been used in the literature in
several contexts, including routing~\cite{KonjevodRX08},
clustering~\cite{AckermannBS10}, nearest neighbour
search~\cite{ColeG06}, and machine learning~\cite{GottliebKK14}.

\section{Coresets}\label{sec:coresets}

The notion of \emph{coreset} has been introduced in \cite{AgarwalHV05}
as a tool for the development of efficient algorithms for optimization
problems on large datasets.  In broad terms, for a given computational
objective, a coreset is a small subset of the input which embodies a
feasible solution whose cost is a good approximation to the cost of
an optimal solution over the entire input.  Coreset constructions have
been successfully developed for the unconstrained diversity
maximization problem~\cite{IndykMMM14,AghamolaeiFZ15,CeccarelloPPU17}.
In fact, these constructions feature an additional
\emph{composability} property, meaning that the construction can be
applied independently to the subsets of an arbitrary input partition
so that the union of the coresets extracted from each subset is itself
a coreset for the entire input. This additional property enables the
development of scalable distributed (e.g., MapReduce) algorithms.
Indeed, the coreset constructions devised in this section for the DMMC
problem are also composable.

Throughout the section, we refer 
to an arbitrary input to a DMMC problem, which is specified by a
set $S$ of size $n$, a matroid $\mathcal{M}=(S, \mathcal{I}(S))$, and
an integer $k \leq \rank(\mathcal{M})$.  The formal definition of
coreset for the problem is the following.
\begin{definition}
For a positive real-valued $\beta \leq 1$, a subset $T \subseteq S$ is a 
\emph{$\beta$-coreset} for
the DMMC problem if
$\diversity_{k, \mathcal{M}} (T)
\geq \beta \diversity_{k, \mathcal{M}} (S)$.
\end{definition}
We aim at  $\beta$-coresets with $\beta$ close to 1.
Before describing how to construct such coresets, we need to
establish some technical results. Let
$\Delta_S = \max_{a, b\in S}(\dist(a, b))$ be  the \emph{diameter} of
$S$. We have:
\begin{fact} \label{fact:diameter}
For $k>1$, there exists an independent set $X \in \mathcal{I}(S)$ of size $k$
containing two points $a,b$ such that
\begin{itemize}
\item
$\dist(a,b) \geq \Delta_S/2$
\item
$\forall c \in X \setminus \{a,b\}$: $\dist(a,c) \geq \Delta_S/4$
or $\dist(b,c) \geq \Delta_S/4$.
\end{itemize}
\end{fact}
\begin{proof}
We first show that there exists an independent set $\{a,b\}$
of two elements at
distance at least $\Delta_S/2$ from one another.  Let $p,q \in S$ be
such that $\Delta_S = \dist(p,q)$.  If $\{p,q\} \in \mathcal{I}(S)$
then the statement clearly holds with $a=p$ and $b=q$.  Otherwise, by
the augmentation property, there must exist a point $r \not\in
\{p,q\}$ such that both $\{p,r\}$ and $\{q,r\}$ are independent sets.
Clearly, by the triangle inequality we have that $\max
\{\dist(p,r),\dist(q,r)\} \ge \Delta_S/2$.  If $\dist(p,r) \ge
\Delta_S/2$ (resp., $\dist(q,r) \ge \Delta_S/2$) the statement holds
with $a=p$ (resp., $a=q$) and $b=r$. The set  $\{a,b\}$ can be
augmented to a set $X$ of size $k$ using the augmentation
property $k-2$ times. Then, the first property stated by the fact
immediately established, and the second property follows
since, by  the triangle inequality, 
any point in $X\setminus \{a,b\}$ must be at distance at
least $\Delta_S/4$ from either $a$ or $b$.
\end{proof}
Observe that each diversity function $\diversity$ listed in
Table~\ref{tab:diversity-notions} is a sum of $f(k)$
distances between points, where $f(k) = {k \choose 2}$,
for max-DMMC, $f(k) = k-1$ for star-DMMC and tree-DMMC,
$f(k) = k$ for cycle-DMMC, and $f(k) =
\lfloor k/2 \rfloor \lceil k/2 \rceil$ for bipartition-DMMC.
A crucial parameter for the analysis of our algorithms
is the \emph{average farness} 
\[
\rho_{S,k} = {\diversity_{k, \mathcal{M}}(S) \over  f(k)}.
\]
The following lemma provides a lower bound to the average farness as a function
of the diameter of the input set, for each of the
diversity functions considered in this paper.
\begin{lemma}\label{lem:diameter}
For $k > 1$, we have that
\[
\rho_{S,k}  \ge 
\left\{
\begin{array}{ll}
\Delta_S/(2k) & \mbox{for sum-DMMC} \\
\Delta_S/(4(k-1)) & \mbox{for star-DMMC} \\
\Delta_S/(2(k-1)) & \mbox{for tree-DMMC} \\
\Delta_S/k & \mbox{for cycle-DMMC} \\
\Delta_S/(2(k+1)) & \mbox{for bipartition-DMMC} \\
\end{array}
\right.
\]
\end{lemma}
\begin{proof}
Consider the independent set $X$ of size $k$, whose existence is proved in
Fact~\ref{fact:diameter}, which contains two points $a,b$ with
$\dist(a,b) \ge \Delta_S/2$ and such that the remaining $k-2$ points
are at distance at least $\Delta_S/4$ from $a$ or $b$, and observe that 
\[
\rho_{S,k} \ge {\diversity(X) \over  f(k)}.
\]
For the sum-DMMC problem, $\diversity(X) \ge (k-1)\Delta_S/4$ since
$\dist(a,b) \geq \Delta_S/2$ and for each $c \in X\setminus \{a,b\}$,
$\dist(a,c)+\dist(b,c) \ge \Delta_S/4$. The bound follows
since, for this diversity function, $f(k) = {k \choose 2}$.  For the
star-DMMC problem, $\diversity(X) \ge \Delta_S/4$ since for any $c \in
X$ there  exists at least one point $u \in X\setminus \{c\}$ such that
$\dist(c,u) \ge \Delta_S/4$.  The bound follows since, for this
diversity function, $f(k) = k-1$.  For the tree-DMMC problem,
$\diversity(X) \ge \Delta_S/2$, since any spanning tree connecting
the points of $X$ includes a path between $a$ and $b$ which
has length at least $\dist(a,b) \ge \Delta_S/2$.  The bound follows
since, for this diversity function, $f(k) = k-1$.  Similarly, for the cycle-DMMC
problem, $\diversity(X) \ge \Delta_S$ since any Hamiltonian cycle connecting the
points of $X$ is made of two edge-disjoint paths between $a$ and $b$
whose aggregate length is at least $2\dist(a,b) \ge \Delta_S$.
The bound follows since, for this diversity function, $f(k) = k$.
For the bipartition-DMMC problem, let $(Q, X\setminus Q)$, with $|Q| = \lfloor k/2\rfloor$, be the 
bipartition of $X$ minimizing the sum $\sum_{u\in Q, v\in X\setminus Q} \dist(u,v)$. We distinguish two 
cases. In case $a$ and $b$ belong to the same subset of the bipartition, then each point $c$ in
the other subset will contribute at least $\dist(a,c)+\dist(b,c) \geq \Delta_S/4$ to the sum, whence 
$\diversity(X) \ge \lfloor k/2\rfloor \Delta_S/4$. Otherwise, assume w.l.o.g.\ that $a\in Q$ and $b\in X\setminus Q$
and observe that these two points contribute at least $\Delta_S/2$ to the sum. Out of the remaining $k-2$ points in $X$,
we can create $\lfloor k/2\rfloor -1$ mutually disjoint pairs $\{u_i, v_i\}$, with $u_i \in Q$ and $v_i \in X\setminus Q$, for $1\leq i<\lfloor k/2\rfloor$. 
We have that either  $\dist(a, v_i) + \dist(b,u_i) \ge \Delta_S/ 4$ or, by the triangle inequality, 
$\dist(u_i,v_i) \geq \dist(a,b)-\dist(a,v_i)-\dist(b,u_i) > \Delta_S/2-\Delta_S/4 = \Delta_S/4$. Thus, pair $\{u_i, v_i\}$ contributes al least 
$\Delta_S/4$ to the sum, hence $\diversity(X) \geq (\lfloor k/2\rfloor -1) \Delta_S/4+\dist(a,b) \ge \lfloor k/2\rfloor \Delta_S/4$.
The bound follows since, for this diversity function, $f(k) = \lfloor k/2\rfloor\lceil k/2\rceil$.
\end{proof}
It is easy to argue that there are instances of the problem for which
the lower bound to $\rho_{S,k}$ is tight, up to constant factors.

Intuitively, a good coreset in
our setting is a set of points that contains, for each independent set in
$\mathcal{I}(S)$, an independent set of comparable diversity. 
The following lemma (which holds for every instantiation of the 
DMMC problem) formalizes this intuition.
\begin{lemma}\label{lem:coreset-requirements}
Let $\epsilon < 1$ be a positive value.  Consider a subset $T
\subseteq S$ such that for each $X \in \mathcal{I}(S)$ of size $k$
there is an injective proxy function $p: X \rightarrow T$ satisfying (i)
$\{p(x) : x \in X\} \in \mathcal{I}(S)$;
and (ii) $\dist(x, p(x)) \le (\epsilon/2)\rho_{S,k}$, for every $x \in X$.  Then, $T$ is a
$(1-\epsilon)$-coreset.
\end{lemma}
\begin{proof}
Let $O \subseteq S$ be an optimal solution to the DMMC instance
and consider the set of proxies $p(O)=\{p(o) : o \in O\} \subseteq T$,
which is an independent set of size $k$ by hypothesis 
and is thus a feasible solution. By the triangle inequality and the
properties of the proxy function, for each pair
$o_1, o_2 \in O$, we have that 
\begin{eqnarray*}
\dist(p(o_1), p(o_2)) & \ge & 
\dist(o_1, o_2)-\dist(o_1, p(o_1))-\dist(o_2, p(o_2)) \\
& \ge &
\dist(o_1, o_2)-\epsilon\rho_{S,k}.
\end{eqnarray*}
It follows that 
\[
\diversity(p(O)) \ge
\diversity(O)-f(k)(\epsilon\rho_{S,k}),
\]
where $f(k)$ denotes, as stated before, 
the number of distances which contribute to $\diversity$.
Since $\diversity(O)=\diversity_{k, \mathcal{M}}(S)=f(k)\rho_{S,k}$,
it follows that 
\[
\diversity(p(O)) \ge (1-\epsilon) \diversity_{k, \mathcal{M}}(S).
\]
The lemma follows, since 
$\diversity_{k, \mathcal{M}}(T) \ge \diversity(p(O))$.
\end{proof}

In the next subsections, we will develop clustering-based
constructions of small coresets meeting the requirements of
Lemma~\ref{lem:coreset-requirements} for partition and transversal
matroids. We will also point out how to extend these constructions to
the case of general matroids, at the expense of a possible blow-up in the
coreset size.  

\subsection{Coreset construction} \label{subsec:construction}

\sloppy
Fix an arbitrary positive constant $\epsilon < 1$, and consider a
$\tau$-\emph{clustering} $(\mathcal{C},Z)$ of the input set $S$, where
$\mathcal{C}=\{C_1, \dots, C_{\tau}\}$ and $Z = \{z_1, \dots,
z_{\tau}\}$, with radius
\begin{equation}\label{eq:proof:partition-radius-bound}
r(\mathcal{C},Z) \le \frac{\epsilon}{4}\rho_{S,k}.
\end{equation}
Observe that such a clustering surely exists, as long as $\tau$ is
large enough, since the trivial $n$-clustering 
where each element of $S$
is a singleton cluster has radius 0. Our coresets are obtained by
selecting a suitable subset from each cluster of $\mathcal{C}$ so that
the properties (i) and (ii) specified in 
Lemma~\ref{lem:coreset-requirements} are satisfied. In particular,
the bound on the clustering radius is functional to establish property (ii).

The effectiveness of this approach relies on the existence of
a clustering with suitably small $\tau$, so that
the resulting coreset size is significantly smaller than $n$.
Although it is not easy to determine a meaningful upper bound to
$\tau$ in the general case, in the next subsection we show that for
the important case of metric spaces of bounded doubling dimension,
$\tau$ is upper bounded by a constant w.r.t. $n$. 

Technically, throughout this section, we work under the hypothesis that
a $\tau$-clustering $(\mathcal{C},Z)$ whose radius satisfies
Equation~\ref{eq:proof:partition-radius-bound} is available, and
postpone the description of its explicit construction to 
Section~\ref{sec:implementation}, since different constructions will
be employed for the different computational settings considered in
this paper.

Below, we describe, separately for each matroid type, how a coresets
can be derived from the $\tau$-clustering $(\mathcal{C},Z)$ of $S$ of
radius $r(\mathcal{C},Z) \le (\epsilon/4)\rho_{S,k}$.

\subsubsection{Partition matroid} \label{subsubsec:partition}
Consider a partition  matroid 
$\mathcal{M}=(S, \mathcal{I}(S))$ with
categories $A_1, \dots, A_h$ and cardinality bounds $k_1,
\dots, k_h$.
We build the coreset $T$ for $S$ as follows.  From each cluster $C_i$
of $\mathcal{C}$ we select a largest independent set $T_i \subseteq
C_i$ of size at most $k$, and let $T = \bigcup_{i=1}^{\tau} T_i$.  The
effectiveness of this simple strategy is stated by the following
theorem.
\begin{theorem}\label{thm:partition-proxy}
The set $T$ computed by the above procedure
from a $\tau$-clustering $(\mathcal{C},Z)$ of $S$ of radius
$r(\mathcal{C},Z) \le (\epsilon/4)\rho_{S,k}$ is a
$(1-\epsilon)$-coreset of size $\BO{k\tau}$ for the DMMC problem.
\end{theorem}
\begin{proof}
First observe that the  bound on the size of $T$ is immediate
by construction. As for the approximation, we now
show that for any
independent set $X \in \mathcal{I}(S)$ with $|X| = k$, there is
an injective function $p: X \rightarrow T$ such that $\{p(x) : x \in
X\} \in \mathcal{I}(S)$ and $\dist(x, p(x)) \le (\epsilon/2)
\rho_{S,k}$. The result will then follow from Lemma~\ref{lem:coreset-requirements}.
Consider the set of clusters as partitioned in two families:
$\mathcal{C}^\ell$ includes those clusters that contain an independent
set of size $k$, whereas $\mathcal{C}^s =
\mathcal{C}-\mathcal{C}^\ell$ includes the remaining clusters that
contain only independent sets of size strictly less than $k$.  For each cluster
$C_i \in \mathcal{C}^s$, consider the independent set $T_i
\subseteq C_i$ included in $T$ by the algorithm.  It is easy to see
that for each category $A_j$, we have $|X \cap C_i \cap A_j|
\le |T_i \cap A_j|$, since
$|T_i| < k$ and that $T_i$ is a largest independent set in $C_i$. 
Therefore, each point $x \in X \cap C_i$ can be
associated with a distinct point $p(x) \in T_i$ belonging to the same
category. Let 
$P = \{p(x) : x \in X \cap (\cup_{C\in \mathcal{C}^s} C)\}$
and note that $P$ is an independent set,
since $X \cap (\cup_{C \in \mathcal{C}^s} C)$ is an
independent set (by the hereditary property of matroids) 
and the number of elements per category is the same
in $X$ and in $P$.
If $P$ has size $k$, then  $X$ has no points in clusters of 
$\mathcal{C}^\ell$ and the lemma is proved. If instead $P$ has
size strictly less than $k$, we consider the clusters in
$\mathcal{C}^\ell$ containing the remaining $k - |P|$ points of $X$.
For each such $C_i$, we expand $P$ to a larger
independent set by adding $n_i=|X \cap C_i|$ elements from $T_i$ using
the augmentation property $n_i$ times, exploiting the fact that $T_i$ is an
independent set of size $k$. These $n_i$ elements of $T_i$ can act as
distinct proxies of the elements in $X \cap C_i$ under function
$p$. After these additions, we obtain an independent set $P$ of $k$
distinct proxies for the elements of $X$.
Since, for each $x \in X$, $p(x) \in P$ is taken
from the same cluster of $\mathcal{C}$, we have that, by the triangle inequality and by
Equation~(\ref{eq:proof:partition-radius-bound}), $\dist(x, p(x)) \le
(\epsilon/2) \rho_{S,k}$.
\end{proof}

\subsubsection{Transversal matroid}
\label{subsubsec:transversal}

The construction of the coreset $T$ is more involved in the case of
transversal matroids. Consider a transversal matroid $\mathcal{M}=(S,
\mathcal{I}(S))$ defined over a family $\mathcal{A}=\{A_1, \ldots,
A_h\}$ of $h$ (not necessarily disjoint) categories.  The fact that
the categories can now overlap complicates the coreset construction,
resulting in slightly larger coresets.  For each cluster $C_i$ of
$\mathcal{C}$ we begin by selecting, as before, a largest independent
set $U_i$ of size at most $k$.  If $|U_i|= k$, then we set $T_i=U_i$.
Instead, if $|U_i| < k$, let $\mathcal{A}' \subseteq \mathcal{A}$ be a
subfamily of categories of the points of $U_i$. We construct 
$T_i$ by augmenting $U_i$ in such a way that for each category $A \in \mathcal{A}'$ there are $\min\{k, |A \cap C_i|\} \leq k$ points of category $A$
in $T_i$. (Observe that a point contributes to the count for all of its
categories in $\mathcal{A}'$) Finally, we let $T = \bigcup_{i=1}^{\tau} T_i$.
\begin{theorem}\label{thm:transversal-proxy}
The set $T$ computed by the above procedure
from a $\tau$-clustering $(\mathcal{C},Z)$ of $S$ of radius
$r(\mathcal{C},Z) \le (\epsilon/4)\rho_{S,k}$ is a
$(1-\epsilon)$-coreset of size $\BO{k^2\tau}$ for the DMMC problem.
\end{theorem}
\begin{proof}
The proof follows the lines of the one of
Theorem~\ref{thm:partition-proxy}. The bound on the size of $T$
follows from the assumption that each point belongs to a constant
number of categories, while for the approximation guarantee of the
coreset it is sufficient to show that for any independent set $X \in
\mathcal{I}(S)$ with $|X| = k$, there is an injective function $p: X
\rightarrow T$ such that $\{p(x) : x \in X\} \in \mathcal{I}(S)$ and
$\dist(x, p(x)) \le (\epsilon/2) \rho_{S,k}$.  Again, we split
$\mathcal{C}$ into $\mathcal{C}^\ell$ (clusters containing an
independent set of size $k$) and $\mathcal{C}^s$ (remaining clusters).
We first consider clusters in $\mathcal{C}^s$.  Fix a cluster $C_i \in
\mathcal{C}^s$ and let $\mathcal{A}' \subseteq \mathcal{A}$ be the
subfamily of categories of the points of $U_i$, and recall that, by
construction, for each $A \in \mathcal{A}'$ there are $\min\{k, |A
\cap C_i|\}$ points of category $A$ in $T_i$.  Let $C_i \cap X =
\{x_1, x_2, \dots, x_m\}$, and let $A_1^X, A_2^X, \dots, A^X_m \in
\mathcal{A}$ be distinct categories that can be matched to $x_1, x_2,
\dots, x_m$. Note that the maximality of $U_i$ implies that $|U_i|
\geq m$.  Without loss of generality, assume that $x_1, \dots, x_j \in
T_i$ while $x_{j+1}, \dots, x_{m} \notin T_i$, for some $0 \leq j \leq
m$, that is, assume that exactly those $j$ points of $X$ are included
in $T_i$.  We initially set the proxies $p(x_1)=x_1, \ldots,
p(x_j)=x_j$.  Then, consider category $A^X_{j+1}$ and observe that
$A^X_{j+1} \in \mathcal{A}'$, otherwise $x_{j+1}$ (matched to
$A_{j+1}^X$) could be added to $U_i$ contradicting its maximality
within $C_i$. Since $x_{j+1}$ was not included in $T_i$, there must be
at least $k \ge m$ elements of $A_{j+1}^X$ in $T_i$.  We can thus
select one such element (distinct from $p(x_1), p(x_2), \dots,
p(x_j)$) as proxy $p(x_{j+1})$, matched to $A_{j+1}^X$.  By repeating
this step for $x_{j+2}, \ldots, x_m$, we obtain an independent set of
proxies for the elements of $X \cap C_i$, each matched to the same
category as its corresponding element. Then, by iterating this
construction over all clusters of $\mathcal{C}^s$ we get an
independent set $P$ of proxies for the elements of $X$ belonging to
such clusters, each matched to the same category as its corresponding
element.  If $P$ has size strictly less than $k$, the remaining
proxies for the elements of $X$ residing in clusters of
$\mathcal{C}^\ell$ can be chosen by reasoning as in the proof of
Theorem~\ref{thm:partition-proxy}, using the augmentation property.
Also, by the bound on the radius of $(\mathcal{C},Z)$, we have
$\dist(x, p(x)) \le (\epsilon/2)\rho_{S,k}$, for each $x \in X$.
\end{proof}

We remark that the $\BO{k^2 \tau}$ bound on the coreset size is a
rather conservative worst-case estimate. In fact, as reported in the
experimental section, it is conceivable that, in practice, much
smaller sizes can be expected. Another important observation is that
the constructions for both the partition and the transversal matroid
yield coresets whose size is independent of $n=|S|$.

\subsubsection{General Matroids}
\label{subsubsec:general}
Consider now a constraint specified as a general matroid
$\mathcal{M}(S, \mathcal{I}(S))$. We can still build a coreset $T$ in
this general scenario from the $\tau$-clustering
$(\mathcal{C},Z)$ of $S$. Specifically, for each cluster $C_i$ of
$\mathcal{C}$ we compute a largest independent set $U_i$ of size at
most $k$.  If $|U_i|=k$, then we set $T_i=U_i$, otherwise we set
$T_i=C_i$. Finally, we let $T = \bigcup_{i=1}^{\tau} T_i$.  Note that,
unlike the case of partition and transversal matroids, this coreset
may potentially grow very large when clusters do not contain large
enough independent sets. However, for small enough values of $\tau$
and $k$ (hence, large cluster sizes) we expect that each
cluster may reasonably contain an independent set of size $k$, hence
an actual coreset size of $\BO{k \tau}$ is conceivable.
We have:

\begin{theorem}\label{thm:general-proxy}
The set $T$ set computed by the above procedure from a
$\tau$-clustering $(\mathcal{C},Z)$ of $S$ of radius $r(\mathcal{C},Z)
\le (\epsilon/4)\rho_{S,k}$ is a $(1-\epsilon)$-coreset for the DMMC problem.
\end{theorem}
\begin{proof}
As in the proofs of Theorems~\ref{thm:partition-proxy}
and~\ref{thm:transversal-proxy}, it is sufficient to determine a
suitable proxy function $p$ for every independent set $X \in
\mathcal{I}(S)$ of size $k$.  We set function $p$ to be the identity
function for those points belonging to clusters $C_i$ such that
$T_i=C_i$.  The proxies for the elements of $X$ residing in the other
clusters can then be chosen through repeated applications of the
augmentation property. Specifically, consider a cluster $C_i$ such
that $T_i \subset C_i$ (hence, $T_i$ is and independent set of size
$k$ by construction) and let $m_i = |X \cap C_i|$. Since $T_i$ is an
independent set of size $k$, we can apply the augmentation property
$m_i$ times to select the $m_i$ proxies $X \cap C_i$ from $T_i$ so
that the union of all proxies is an independent set. 
Also, by the bound on the radius of
$(\mathcal{C},Z)$, we have $\dist(x, p(x)) \le
(\epsilon/2)\rho_{S,k}$, for each $x \in X$, hence the thesis
follows from Lemma~\ref{lem:coreset-requirements}.
\end{proof}

\subsection{Relating cluster granularity to doubling dimension}

In this subsection, we show that for datasets of bounded doubling
dimension, very coarse clusterings of small radius exist, and 
discuss how this feature can be exploited to obtain coresets
of small size for the DMMC problem.

\begin{theorem}\label{thm:coreset-doubling-dimension}
Let $S$ be an $n$-point dataset of doubling dimension $D$.
For any integer $\tau$, with $1 \leq \tau \leq n$,
the minimum radius of any $\tau$-clustering of $S$, denoted with $r_\tau^*(S)$, is
\[
r_\tau^*(S) \leq \frac{2\Delta_S}{\tau^{1/D}}.
\]
\end{theorem}
\begin{proof}
Observe that the whole set $S$ is contained in the ball of radius 
$\Delta_S$ centered at any element of $S$.  By applying
the definition of doubling dimension $i$ times, starting from such a ball, 
we can cover $S$ with $2^{iD}$ balls of radius at most
$\Delta_S/2^i$. Let $j$ be such $2^{jD} \leq \tau < 2^{(j+1)D}$. 
The theorem follows since
the minimum radius of any $\tau$-clustering of $S$ is upper bounded
by the minimum radius of any $2^{jD}$-clustering of $S$, which is in
turn upper bounded by $\Delta_S/2^j \leq 2\Delta_S/\tau^{1/D}$.
\end{proof}
In order to appreciate the relevance of the above theorem, consider an
algorithm $\mathcal{A}$ that, given a target number of clusters $\tau$,
returns a $\tau$-clustering for $S$ whose radius $r_\tau^{\mathcal{A}}(S)$
is a factor at most $\sigma >1$ larger than the
minimum radius attainable by any $\tau$-clustering for $S$, that is $r_\tau^{\mathcal{A}}(S) \le \sigma\cdot r_\tau^*(S)$. The theorem
implies that, by setting $\tau = (32\sigma k/\epsilon)^D$, Algorithm
${\mathcal{A}}$ returns a $\tau$-clustering for $S$ with
radius at most
\[
r_\tau^{\mathcal{A}}(S) \le 
\sigma r_\tau^*(S) \le
\sigma \frac{2\Delta_S}{\tau^{1/D}} 
\le
\frac{\epsilon \Delta_S}{16k} 
\le
\epsilon \frac{\rho_{S,k}}{4},
\]
for any of the DMMC instantiations in
Table~\ref{tab:diversity-notions}, where the last inequality follows
from Lemma~\ref{lem:diameter}. Observe that such a small-radius
clustering can be used as the base for the coreset constructions illustrated
in the previous subsection, and its limited granularity ensures that
the coresets have size independent of $n$ for the partition and
transversal matroids. This feature will be crucial for obtaining
efficient sequential, distributed, and streaming algorithms for all
instantiations of the DMMC problem considered in this paper.
Note that this approach requires the knowledge of the
doubling dimension $D$ in order to set $\tau$ properly. 
However, in the next section, we will devise 
implementations of the constructions that return coresets of
comparable quality and size without knowledge of
$D$. This is a very desirable feature for practical purposes
since the doubling dimension of a dataset is hard to estimate.

\section{DMMC algorithms}
\label{sec:implementation}
In this section, we will devise efficient sequential, MapReduce and
Streaming algorithms for the instantiations of the DMMC problem
defined in Table~\ref{tab:diversity-notions}.  At the core of these
algorithms are efficient implementations of the coreset constructions
presented in the previous section.  Consider any of the instantiations
of the DMMC problem and an arbitrary input specified by a set $S$ of
size $n$, a matroid $\mathcal{M}=(S, \mathcal{I}(S))$, and an integer
$k \leq \rank(\mathcal{M})$.  Our general approach is to extract a
$(1-\epsilon)$-coreset $T$ from $S$ and then run the best available
sequential algorithm on $T$. 
In Subsections~\ref{sec:seq-implementation},
\ref{sec:mr-implementation}, and \ref{sec:stream-implementation}, we
present the implementations of the $(1-\epsilon)$-coreset construction in the
sequential, MapReduce and Streaming setting, respectively.  Finally,
in Subsection~\ref{sec:final-algorithm} we will show how to employ
these constructions to yield the final algorithms in the various
setting and analyze their performance.

As in previous works, we assume that
constant-time oracles are available to compute the distance between
two elements of $S$ and to check whether a subset of $S$ is an
independent set \cite{AbbassiMT13}.

\subsection{Sequential coreset construction}
\label{sec:seq-implementation}
Our sequential implementation of the $(1-\epsilon)$-coreset
construction presented in the previous section, dubbed {\sc SeqCoreset}
(see Algorithm~\ref{alg-sequential} for the pseudocode) leverages the
well-known 2-approximate clustering algorithm \algo{gmm} mentioned
before \cite{Gonzalez85}. For a given input $S$, \algo{gmm} determines
the set $Z$ of cluster centers in $|Z|$ iterations, by initializing
$Z$ with an arbitrary element of $z_1\in S$, and then iteratively adding to
$Z$ the element in $S$ of maximum distance from the current $Z$. The
algorithm can be instrumented to maintain the set of clusters
$C=\{C_z: z\in Z\}$ centered at $Z$, with each element of $S$
assigned to the cluster of its closest center, and the radius of such
a clustering.  Each iteration of \algo{gmm} can be easily implemented
in time linear in $n$ (see Procedure \algo{gmm-iteration} in
Algorithm~\ref{alg-sequential}) and, as proved in the original paper
\cite{Gonzalez85}, the clustering resulting at the end of the $i$-th
iteration has radius which is no more than twice as large as the
minimum radius of any $i$-clustering.

Since the second center $z_2$ selected by \algo{gmm} is the farthest point
from the first (arbitrarily selected) center $z_1$, it is easy to see that
the distance $\delta=d(z_1,z_2)$ between these first two centers is such that
$\Delta_S/2 \leq \delta \leq \Delta_S$.  In order to build the
$(1-\epsilon)$-coreset $T$, we run \algo{gmm} for a number $\tau$ of
iterations sufficient to reduce the radius of the clustering to a
value at most $\epsilon \delta/(16k)$. Once the clustering is
computed, for each cluster $C_z$ a largest independent set $U_z
\subseteq C_z$ of size at most $k$ is determined.  For the partition
matroid, the coreset $T$ is obtained as the union of the $U_z$'s.  For
the transversal matroid, $T$ is obtained by first augmenting each
$U_z$ in such a way that for each category $A$ of a point of $U_i$,
the augmented set contains $\min\{k, |A \cap C_i|\}$ points of $A \cap
C_z$, and then taking the union of these augmented sets. Finally for
all other matroids, $T$ is obtained by first augmenting each
independent set $U_z$ of size $|U_z| <k$ to the entire cluster $C_z$,
and then again taking the union of the $U_z$'s (see Procedure
\algo{extract} in Algorithm~\ref{alg-sequential}).
 
\begin{algorithm}
\caption{{\sc SeqCoreset}$(S,k,\epsilon)$}
\label{alg-sequential}
  \DontPrintSemicolon
Let $S=\{x_1,x_2, \ldots, x_n\}$\;
\Let{$z_1$}{$x_1$}\; 
\Let{$z_2$}{$\argmax_{x\in S}\{d(z_1,x)\}$}\;
\Let{$\delta$}{$d(z_1,z_2)$}\;
\Let{$Z$}{$\{z_1,z_2\}$}\;
\lFor{$i\in \{1,2\}$}{%
\Let{$C_{z_i}$}{$\{x\in S: z_i =\argmin_{z'\in Z}\{d(x,z')\} \}$}
}
\Let{$C$}{$\{C_{z_1},C_{z_2}\}$}\;
\lWhile{$\left(r(C,Z)>\epsilon\delta/(16k)\right)$}{%
\Let{$(C,Z)$}{{\sc gmm-iteration}$(S,Z)$}
}
\lFor{$z\in Z$}{%
\Let{$U_z$}{{\sc extract}$(C_z,k)$}
}
\Return $T=\cup_{z \in Z}U_z$\;

\vspace*{0.3cm}
{\bf procedure} {\sc gmm-iteration}$(S,Z)$\;
\Let{$y$}{$\argmax_{x\in S}\{\min_{z\in Z}\{d(x,z\}\}$}\; 
\Let{$Z$}{$Z\cup\{y\}$}\;
\lFor{$z\in Z$}{%
\Let{$C_z$}{$\{x\in S: z = \argmin_{z'\in Z}\{d(x,z')\}\}$}\;
\Let{$C$}{$\{C_z: z\in Z\}$}
}
\Return $(C,Z)$

\vspace*{0.3cm}
{\bf procedure} {\sc extract}$(C,k)$\;
\Let{$U$}{maximal independent set in $C$ of size $\leq k$}\;
\lIf{\rm ($(|U|=k)$ $\vee$ (\emph{matroid type = partition}))}{\Return $U$}
\Else{
\Switch{(matroid type)}{
\Case{transversal}{
\If{$(\exists \mbox{ \rm Category $A$ of $x\in U$}: |A \cap U|<k)$}
{
add to $U$ extra points from $C$ so to have $\min\{k, |A\cap C|\}$ points of Category $A$ in $U$
}
}
\lCase{other}{
\Let{$U$}{$C$}
}
}
}
\Return $U$
\end{algorithm}

\noindent  We have:
\begin{theorem}
\label{thm:seqcoreset}
Let $\epsilon <1$ be an arbitrary positive constant.  The above
algorithm computes a $(1-\epsilon)$-coreset $T$ for the DMMC problem
in time $\BO{n\tau}$.  For the partition (resp., transversal)
matroid, $T$ has size $\BO{k \tau}$ (resp., $\BO{k^2 \tau}$). If the
set $S$ has constant doubling dimension $D$, then $\tau =
\BO{(k/\epsilon)^D}$.
\end{theorem}
\begin{proof}
It is easy to see that the distance $\delta$ between the first two centers
selected by \algo{gmm} is such that $\Delta_S/2 \le \delta \le \Delta_S$,
thus, by Lemma~\ref{lem:diameter}, the radius of 
the $\tau$-clustering is
$\epsilon \delta/(16k) \le \epsilon \Delta_{S}/(16k) \leq
\epsilon\rho_{S,k}/4$. The fact that $T$ is $(1-\epsilon)$-coreset 
for the DMMC problem with the stated sizes then
follows by Theorems~\ref{thm:partition-proxy} and
~\ref{thm:transversal-proxy}.
As for the running time, the cost of the
\algo{gmm} algorithm is $\BO{n\tau}$, while, the subsequent extraction of the
coreset can be accomplished using a total of $\BO{n}$ invocations of
the independent-set oracle to determine a largest independent set
in each cluster, exploiting the augmentation property, and additional
$\BO{nk} = \BO{n\tau}$ operations on suitable dictionary structures to
determine the extra elements for each category (in the sole case of
the transversal matroid).  Finally, the bound on the cluster granularity can be
established by noticing that for $\tau = (128k/\epsilon)^D$,
Theorem~\ref{thm:coreset-doubling-dimension} together with the fact that
\algo{gmm} is a 2-approximate clustering algorithm, implies that
after this many iterations of \algo{gmm} the
radius of the clustering is at most
\[
\frac{4 \Delta_S}{(\tau)^{1/D}} 
\le 
\frac{\epsilon \delta}{16k}.
\]
Note that for constant $D$, $\tau = \BO{(k/\epsilon)^D}$.
\end{proof}
It is conceivable that the large constants involved in the coreset
sizes are an artifact of the analysis since the experiments reported
in Section~\ref{sec:experiments}, show that much smaller coresets yield very accurate
solutions. We also wish to stress that, thanks to the incremental
nature of \algo{gmm}, the coreset construction needs not know the
doubling dimension $D$ of the metric space in order to attain the
desired bound on $\tau$.

\subsection{MapReduce coreset construction}
\label{sec:mr-implementation}
A MapReduce (MR) algorithm executes as a sequence of \emph{rounds}
where, in a round, a multiset of key-value pairs is transformed into a
new multiset of pairs through two user-specified map and reduce
functions, as follows: first the map function is applied to each
individual pair returning a set of new pairs; then the reduce function
is applied independently to each subset of pairs having the same key,
again producing a set of new pairs.  Each application of the reduce
function to a subset of same-key pairs is referred to as
\emph{reducer}.  The model is parameterized by the total memory
available to the computation, denoted with $M_T$, and by the maximum
amount of memory locally available to run each map and reduce
function, denoted with $M_L$.  The typical goal for a MR algorithm is
to run in as few rounds as possible while keeping $M_T$ (resp., $M_L$)
linear (resp., substantially sublinear) in the input size
\cite{Dean2004,KarloffSV10,PietracaprinaPRSU12}.

To obtain MR implementations of our $(1-\epsilon)$-coreset constructions we will
crucially exploit an additional property of these constructions, known
as \emph{composability} \cite{IndykMMM14}. Formally, composability
ensures that a $(1-\epsilon)$-coreset for a set $S$ can be obtained as
the union of $(1-\epsilon)$-coresets extracted from each subset of any
given partition of $S$.  Therefore, a coreset $T$ can be built in one
MR round as follows.  First, the set $S$ is partitioned evenly but
arbitrarily into $\ell > 0$ disjoint subsets $S_1, \dots, S_\ell$,
through a map function, where $\ell$ is a design parameter (corresponding to the degree of parallelism) to be set
in the analysis.  Then, each $S_i$ is assigned to a distinct reducer,
which builds a $(1-\epsilon)$-coreset $T_i$ for $S_i$ based on a
$\tau_i$-clustering of radius at most $\epsilon \delta_i/(16k)$, by running 
 {\sc SeqCoreset}$(S_i,k,\epsilon)$ (see Algorithm~\ref{alg-sequential}),
where $\delta_i \in [\Delta_{S_i}/2,\Delta_{S_i}]$ now represents the
distance between the first two centers selected by the algorithm. 
Coreset $T$ is simply the union of the $T_i$'s.

\begin{theorem}
\label{thm:MRcoreset} 
Let $\epsilon <1$ be an arbitrary positive constant and let $\tau =
\sum_{1 \leq i \leq \ell} \tau_i$.  The above 1-round MR algorithm
computes a $(1-\epsilon)$-coreset $T$ for the DMMC problem with memory
requirements $M_T=\BO{n}$ and $M_L =\BO{n/\ell}$.  For the partition
(resp., transversal) matroid, $T$ has size $\BO{k \tau}$ (resp.,
$\BO{k^2 \tau}$).  If the set $S$ has constant doubling dimension $D$,
then $\tau = \BO{\ell (k/\epsilon)^D}$.
\end{theorem}

\begin{proof}
The coreset $T=\bigcup_{i=1,\ell}T_i$
computed by the algorithm can be regarded as being derived from a
$\tau$-clustering of $S$
of radius at most
\[
\epsilon {\max_{i=1,\ell} \Delta_{S_i} \over 16k}
\leq
\epsilon {\Delta_S \over 16k} \leq 
\epsilon {\rho_{S,k} \over 4},
\]
where the last inequality follows by Lemma~\ref{lem:diameter}.
Hence, by
Theorems~\ref{thm:partition-proxy} and ~\ref{thm:transversal-proxy},
$T$ is $(1-\epsilon)$-coreset for $S$.
The bound on the total and local memory is immediate, since the
sequential algorithm used to extract each $T_i$ runs in linear space.
For what concerns the bound on $\tau$, consider a generic subset $S_i$.
It is easy to adapt the proof of Theorem~\ref{thm:coreset-doubling-dimension} 
to show that for $\tau_i = (256k/\epsilon)^D$,
there is a $\tau_i$-clustering of $S_i$ of radius
$\epsilon \Delta_{S_i}/(128k)$, whose centers are points of $S$ which
may not belong to $S_i$. By recentering each cluster on a point
of $S_i$, the radius at most doubles. Therefore, after $\tau_i$ 
iterations \algo{gmm} returns a clustering of radius at most
\[
\frac{4 \epsilon \Delta_{S_i}}{128k} 
\le 
\frac{\epsilon \delta_i}{16k}.
\]
Hence, $\tau = \sum_{1 \le i \le \ell} \tau_i \le \ell (256k/\epsilon)^D$,
thus for constant $D$, $\tau = \BO{\ell(k/\epsilon)^D}$.
\end{proof}
As will be illustrated in Subsection~\ref{sec:final-algorithm}, in a
second round the coreset $T$ can be gathered in a single reducer which
will then extract the final solution by running on $T$ the best available
sequential approximation algorithm.  The degree of parallelism $\ell$
can then be fixed in such a way to balance the local memory
requirements of both rounds. However, if this balancing results in a
large value of $\ell$ (possibly a
function of $n$), since the approximation algorithm used to extract the
final solution is computationally intesive, the work in the second
round could easily grow too large.  To circumvent this problem, the slow
approximation algorithm can be run on a smaller coreset $T'$ of size
independent of $n$, which can be computed from $T$ using our
sequential coreset construction, at the expense of an extra
$1-\epsilon$ factor in the final approximation ratio.

\subsection{Streaming coreset construction}
\label{sec:stream-implementation} 

In the streaming setting \cite{HenzingerRR98} the computation is
performed by a single processor with a small-size working memory, and
the input is provided as a continuous stream of items which is usually
too large to fit in the working memory. Typically, streaming
strategies aim at a single pass on the input but in some cases few
additional passes may be needed.  Key performance indicators are the
size of the working memory and the number of passes.

In this subsection, we describe a 1-pass streaming algorithm, dubbed
{\sc StreamCoreset} which implements the $(1-\epsilon)$-coreset
construction for the DMMC problem described in
Section~\ref{subsec:construction}.  At the core of the algorithm
is the computation of a set $Z$ of centers implicitly defining a
$|Z|$-clustering of radius at most $\epsilon \Delta_S/(16k)$, which is 
in turn upper bounded by $\epsilon \rho_{S,k}/4$ by  Lemma~\ref{lem:diameter}.
Specifically, $Z$ is obtained by combining a center
selection strategy akin to the one presented in
\cite{CharikarCFM04}, with a progressive estimation of $\Delta_S$,
needed to obtain the desired bound on the cluster radius.  As
$Z$ is computed, a set of extra points, dubbed \emph{delegates}, 
is selected from each of the clusters induced by $Z$ so that, at the end of the stream, 
all selected points provide the desired coreset for each matroid type.

Algorithm {\sc StreamCoreset} works as follows (see
Algorithm~\ref{alg-streaming} for the complete pseudocode). The algorithm maintains the following variables: an estimate
$R$ for the diameter of the first $i$ points of the stream, the
current set of centers $Z$, and a set of delegates $D_z$ for each
center $z \in Z$. For $1
\leq i \leq n$, let $x_i$ denote the $i$-th point of $S$ in the
stream. Initally, we set $R=d(x_1,x_2)$, $Z=\{x_1,x_2\}$,
and $D_{x_i}=\{x_i\}$, for $i=1,2$.  Let $c>0$ be a suitable constant
which will be set in the analysis. For $i\geq 3$, the processing of
point $x_i$ is performed as follows.  Let $z$ be the center in $Z$
closest to $x_i$.  If $d(x_i,z) > 2\epsilon R/(ck)$, then $x_i$ is
added to $Z$ and a new set of delegates $D_{x_i} = \{x_i\}$ is
created. Otherwise, $x_i$ triggers an update of the delegate set $D_z$,
carried out by calling procedure {\sc Handle}$(x_i,z,D_z)$, whose
actions vary with the matroid type and will be described later.  Also,
if $d(x_i,x_1) > 2R$, the diameter estimate is updated by setting
$R=d(x_i,x_1)$. Finally, if this latter update occurs, a restructuring
of $Z$ takes place, where $Z$ is shrunk to a maximal subset $Z'$ of
centers at distance greater than $\epsilon R/(ck)$ from one another,
and the delegate set of each discarded center $z \in Z-Z'$ is
``merged'' into the delegate set of its closest center $z' \in Z'$ by
invoking {\sc Handle}$(x,z',D_{z'})$ on each $x \in D_z$.

Given a point $x$, a center $z \in Z$, and its associated delegate set
$D_z$, procedure {\sc Handle}$(x,z,D_z)$ first checks whether $D_z$ is
an independent set of size $k$ and discards $x$ if this is the case.
Otherwise, $D_z$ is updated differently according to the matroid type.
For the partition matroid, $x$ is added to $D_z$ only if $D_z \cup
\{x\}$ is an independent set (clearly of size at most $k$). For the
transversal matroid, $x$ is added to $D_z$ only if one of the
categories of $x$ is still short of $k$ delegates in $D_z$, while
for any other matroid type, $x$ is always added to $D_z$. In all
cases, after $x$ is added to $D_z$, if $D_z$ contains an independent
set $D'$ of size $k$, $D'$ becomes the new $D_z$ and all other points
in $D_z-D'$ are discarded.

\begin{algorithm}
\caption{{\sc StreamCoreset}$(S,k,\epsilon,c)$}
\label{alg-streaming}
  \DontPrintSemicolon
Let $S=x_1,x_2, \ldots$\;
\Let{$R$}{$d(x_1,x_2)$}\;
\Let{$Z$}{$\{x_1,x_2\}$}\;
\Let{$D_{x_1}$}{$\{x_1\}$}; \Let{$D_{x_2}$}{$\{x_2\}$}\;
\For{$(i \geq 3)$}{
\Let{$z$}{$\argmin_{w \in Z}\{d(x_i,w)\}$}\;
\If{$(d(x_i,z) > 2\epsilon R/(ck))$}
{
\Let{$Z$}{$Z \cup \{x_i\}$};
\Let{$D_{x_i}$}{$\{x_i\}$}
}
\lElse{{\sc Handle}$(x_i,z,D_z)$}
\If{$(d(x_i,x_1) > 2R)$}{
\Let{$R$}{$d(x_i,x_1)$}\;
\Let{$Z'$}{maximal subset of $Z$ such that \\
\hspace*{0.8cm} $\forall u \neq v \in Z': d(u,v)>\epsilon R/(ck)$}\;
\For{$(z \in Z-Z')$}{
\Let{$z'$}{$\argmin_{w \in Z'}\{d(z,w)\}$}\;
\lFor{$(x \in D_z)$}{{\sc Handle}$(x,z',D_{z'})$}
}
}
}
\Return $T=\cup_{z \in Z}D_z$\;

\vspace*{0.3cm}
{\bf procedure} {\sc Handle}$(x,z,D_z)$\;
\lIf{$((|D_z|=k) \wedge (D_z \in {\mathcal{I}}(S)))$}{discard $x$}
\Else{
\Switch{(matroid type)}{
\Case{partition}{
\lIf{$(D_z \cup \{x\} \in {\mathcal{I}}(S))$}
{\Let{$D_z$}{$D_z \cup \{x\}$}}
\lElse{discard $x$}
}
\Case{transversal}{
\If{$(\exists \mbox{ \rm Category $A$ of $x$}: |A \cap D_z|<k)$}
{
\Let{$D_z$}{$D_z \cup \{x\}$}\;
\If{$(\exists D' \! \subseteq \! D_z:D' \in {\mathcal{I}}(S) \wedge |D'|=k)$} 
{\Let{$D_z$}{$D'$}\;
discard all other points}
}
\lElse{discard $x$}
}
\Case{other}{
\Let{$D_z$}{$D_z \cup \{x\}$}\;
\If{$(\exists D' \! \subseteq \! D_z:D' \in {\mathcal{I}}(S) \wedge |D'|=k)$} 
{\Let{$D_z$}{$D'$}\;
discard all other points}
}
}
}
\end{algorithm}
In order to analyze the algorithm, we need to make some preliminary
observations and introduce some notation.  Let $Z_i$ be the set of centers
after processing the $i$-th point of the stream. For each point $x_j$
we define the sequence of centers $z^{(j)}_{x_j}, z^{(j+1)}_{x_j},
\ldots, z^{(n)}_{x_j}$, with $z^{(i)}_{x_j} \in Z_i$ for each $i\geq j$, as follows:
$z^{(j)}_{x_j}$ is the center in $Z_j$ closest to $x_j$ (possibly
$x_j$ itself); for $i > j$, if $z^{(i-1)}_{x_j}\in Z_i$  then $z^{(i)}_{x_j} = z^{(i-1)}_{x_j}$, otherwise
$z^{(i)}_{x_j}$ is the center of $Z_i$ closest to $z^{(i-1)}_{x_j}$.
For each $x_j$ we say that $z^{(n)}_{x_j}$ is its \emph{reference center}.
Let $Z=Z_n$ be the set of centers at the end of the algorithm. For
every $z \in Z$, let $C_z$ denote the set all points $x \in S$ for
which $z$ is the reference center (note that $z$ is the reference center for itself).  Clearly, $({\mathcal{C}}=\{C_z : z
\in Z\},Z)$ is a $|Z|$-clustering of $S$. Also, by observing that every
time a delegate $x$ is transferred from a set $D_z$ to a set $D_{z'}$ its
reference center becomes $z'$, it is immediate to conclude that
at the end of the algorithm, $D_z \subseteq C_z$ for each $z \in Z$.

\begin{lemma} \label{lem:zclustering}
Consider the execution of {\sc StreamCoreset}$(S,k,\epsilon,c)$ 
for an arbitrary positive constant $\epsilon$ and $c=32$.
The $|Z|$-clustering $({\mathcal{C}},Z)$ defined above
has radius $r({\mathcal{C}},Z) < (\epsilon/4) \rho_{S,k}$.
Also, if the set $S$ has constant doubling dimension $D$, 
then $|Z| = \BO{(k/\epsilon)^D}$. 
\end{lemma}

\begin{proof}
Let $S_i = \{x_1, x_2, \ldots, x_i\}$, and let $R_i$ be the value of
variable $R$ after processing $x_i$. Analogously, as defined above,
$Z_i$ is the value of $Z$ after processing $x_i$.  To prove the first
part of the lemma, we show that the following invariants are
inductively maintained for every $i \geq 2$:
\begin{enumerate}
\item
$\Delta_{S_i}/4 \leq R_i \leq \Delta_{S_i}$;
\item
For any two distinct centers $u,v \in Z_i$, $d(u,v) > \epsilon R_i/(ck)$;
\item
For every $h \leq i$, it holds
$d(x_h,z^{(i)}_{x_h}) < 2\epsilon R_i/(ck)$.
\end{enumerate}
Initially ($i=2$) the invariants trivially hold for any $\epsilon <1$
and $c \geq 1$. Assume now that the invariants hold after processing
$x_{i-1}$ and consider the processing of $x_i$.  Consider first
Invariant~1, and note that
\begin{eqnarray*}
\Delta_{S_i} & = & \max\{\Delta_{S_{i-1}}, \max_{j<i} \{d(x_i,x_j)\}\} \\
& \leq & \max\{4R_{i-1}, \max_{j<i} \{d(x_i,x_1)+d(x_1,x_j)\}\},
\end{eqnarray*}
where the last passage follows from Invariant 1 at step $i-1$
and from the triangle inequality.
We distinguish two cases.  In case $R_i=R_{i-1}$, we have
that $d(x_i,x_1) \leq 2R_i$ and, moreover, for any
$j < i$, $d(x_1,x_j) \leq 2R_i$ as well,
since the algorithm enforces that $d(x_1,x_j) \leq
2R_j$, and $R_j \leq R_i$.
Therefore, by the above relation it follows that 
$\Delta_{S_i}/4 \leq R_i$. Also, 
$R_i=R_{i-1} \leq \Delta_{S_{i-1}} \leq \Delta_{S_i}$.
If instead $R_i\neq R_{i-1}$, then $R_i=d(x_i,x_1)$, and an easy induction shows that
$d(x_i,x_1) = \max_{j \leq i}\{d(x_j,x_1)\}$. 
Let $\Delta_{S_i} = d(x_r,x_t)$ for some $1 \leq r, t \leq i$. Then,
$\Delta_{S_i} \leq d(x_r,x_1)+d(x_1,x_t) \leq 2d(x_i,x_1)=2R_i$. 
Since, trivially, $d(x_i,x_1) \leq \Delta_{S_i}$, it follows
that
\[ 
\Delta_{S_i}/4 \leq \Delta_{S_i}/2 \leq R_i \leq \Delta_{S_i},
\]
which ensures that Invariant 1 holds. Invariant 2 is explicitly enforced
by the algorithm. As for Invariant 3, 
we distinguish again two cases.  In case $R_i=R_{i-1}$,
then it is immediate to see that the invariant still holds. 
Otherwise, we have that $R_i = d(x_i,x_1)>2R_{i-1}$. Now, for
any point
$x_h$ with $h \leq i$, if $z^{(i)}_{x_h}=z^{(i-1)}_{x_h}$
then $d(x_h,z^{(i)}_{x_h}) < 2\epsilon R_i/(ck)$ by the
fact that the invariant holds at step $i-1$ and that $R_{i-1} <
R_i$. If instead, $z^{(i)}_{x_h} \neq z^{(i-1)}_{x_h}$ then, by the
triangle inequality, 
\begin{eqnarray*}
d(x_h,z^{(i)}_{x_h}) & \leq & 
d(x_h,z^{(i-1)}_{x_h})+d(z^{(i-1)}_{x_h},z^{(i)}_{x_h}) \\
& \leq & \epsilon (2R_{i-1}+R_i)/(c k)  \\
& < & 2\epsilon R_i/(ck).  \\
\end{eqnarray*}

By fixing $c=32$ we have that the  invariants imply that
at the end of the algorithm, the distance between any $x_j \in S$ 
and its reference center in $Z=Z_n$ is
\[
d(x_j,z^{(n)}_{x_j}) <  2\epsilon R_n/(c k) = \epsilon \Delta_{S}/(16k).
\]
The stated bound on the radius of the clustering
$({\mathcal{C}},Z)$ follows since $\Delta_{S}/(4k) \leq \rho_{S,k}$.

For what concerns the bound on $|Z|$, let $\tau$ be the smallest integer such
that the radius of an optimal $\tau$-clustering of $S$ is
at most $\epsilon \Delta_S/(256k)$. 
By Theorem~\ref{thm:coreset-doubling-dimension} it follows that
\[
\tau \leq 
\left\lceil
\left({512 k \over \epsilon}\right)^D
\right\rceil.
\]
We now prove that $|Z| \leq \tau$. 
If this is were not the case,  by the pigeonhole principle there would be two distinct points $z_1, z_2 \in Z$ belonging to the same
cluster of an optimal $\tau$-clustering. Then, by the triangle inequality
and Invariant 1, we would have that
\[
d(z_1,z_2) \leq {2\epsilon \Delta_S \over 256k }
\leq {\epsilon R_n \over 32k },
\]
thus contradicting Invariant 2.
\end{proof}

\begin{theorem}
\label{thm:streamcoreset}
Consider the execution of {\sc StreamCoreset}$(S,k,\epsilon,c)$ for an
arbitrary positive constant $\epsilon$ and  $c=32$.  Then the returned set  
$T=\bigcup_{z\in Z}D_z$ is a $(1-\epsilon)$-coreset for the DMMC problem. 
The algorithm performs a single pass on the stream and uses a working memory of size $O(|T|)$. If the set $S$ has constant doubling dimension $D$, then we have
that $|T|=\BO{k(k/\epsilon)^D}$, for the partition matroid,
and $|T|=\BO{k^2(k/\epsilon)^D}$, for the transversal matroid.
\end{theorem}
\begin{proof}
Let $Z$ be the final set of centers computed by the algorithm and let
$({\mathcal{C}}=\{C_z : z \in Z\},Z)$ be the $|Z|$-clustering of $S$.
Lemma~\ref{lem:zclustering} shows that the radius of this clustering
is at most $(\epsilon/4) \rho_{S,k}$.  Recall that in
Subsections~\ref{subsubsec:partition}, \ref{subsubsec:transversal},
and \ref{subsubsec:general} we showed how to construct
$(1-\epsilon)$-coresets for the various matroid types starting from a
any clustering of radius at most $(\epsilon/4) \rho_{S,k}$. We now
show that at the end of Algorithm {\sc StreamCoreset} each set $D_z
\subseteq C_z$ complies with the requirements of those
constructions. Then, the fact that $T = \bigcup_{z \in Z}D_z$ is a
$(1-\epsilon)$-coreset will follow immediately from
Theorems~\ref{thm:partition-proxy}, \ref{thm:transversal-proxy}, and
\ref{thm:general-proxy}.

For the partition matroid, consider an arbitrary delegate set $D_z$
and observe that the algorithm ensures that it is an independent
set. Hence, if $|D_z| = k$, then $D_z$ clearly complies with the
construction requirements. If instead $|D_z| < k$, consider a point $x_j$ whose
associated sequence of centers is $\{z_{x_j}^{(i)} \; : \; j \leq i
\leq n\}$, and whose reference center is $z_{x_j}^{(n)} = z$.  Let
$A$ be the category of $x_j$ and let $k_A$ be the cardinality
constraint on $A$.  It is easy to see that if $x_j$ is discarded at
some time $i \geq j$ then every $D_{z_{x_j}^{(\ell)}}$, with $\ell
\geq i$, contains $k_A$ elements of $A$. Thus, $x_j$ cannot contribute
to an independent set of $C_z$ larger than $D_z$.

For the transversal matroid, consider as before an arbitrary $D_z$.
Again, if $D_z$ is an independent set of size $k$, then we are
done. Otherwise, let $Q$ be a largest independent set in $D_z$, (thus
$|Q| < k$).  We first show that $Q$ is also a largest independent set
in $C_z$. If this were not the case, by the augmentation property
there would exist a point $x_j \in C_z-D_z$ such that $Q \cup \{x_j\}$
is an independent set. Let $A$ be one of the categories that can be
associated with $x_j$ to provide a matching witnessing the
independence of $Q \cup \{x_j\}$. Let also $\{z_{x_j}^{(i)} \; : \; j
\leq i \leq n\}$ be the sequence of centers associated with $x_j$
(hence, $z_{x_j}^{(n)} = z$).  It is easy to see that if $x_j$ is
discarded at some time $i \geq j$, then every $D_{z_{x_j}^{(\ell)}}$,
with $\ell \geq i$, must contain $k$ elements of $A$.  Therefore,
$D_z$ contains $k$ points of $A$ and one such point can be added to
$Q$ yielding a larger independent set, which contradicts the
maximality of $Q$. To conclude the proof, we need to show that $D_z$
contains at least $\min\{k,|A \cap C_z|\}$ points for each category
$A$ of a point in $Q$. By contradiction, if $|A \cap D_z| < \min\{k,|A
\cap C_z|\}$ for one such category $A$, there would exist a point $x_j
\in A \cap C_z$ which has been discarded by the algorithm at some time
$i \geq j$ when $A$ had less than $k$ points in
$D_{z_{x_j}^{(i)}}$. This is in contrast with the workings of the
algorithm.

The case of the general matroid follows from an easy inductive
argument, which shows that at the end of the algorithm each $D_z
\subseteq C_z$ is either $C_z$ or an independent set of $C_z$ of size
$k$.

Finally, for what concerns the size of $T$,
Lemma~\ref{lem:zclustering} ensures that $|Z| =\BO{(k/\epsilon)^D}$.
For the case of the partition matroid the claimed bound follows the
fact that at any time a delegate set $D_z$ contains at most $k$
points. Instead, for the transversal matroid, the assumption that each
point belongs to at most a constant number of categories (say $\gamma
\in \BO{1}$) and the fact that a point $x$ is added to a delegate set
$D_z$ only if one of the categories of $x$ has less than $k$
representatives in $D_z$, imply that at any time $|D_z| < \gamma k^2$
since otherwise $D_z$ would contain an independent set of size $k$ and
the algorithm would retain only such an independent set.
\end{proof}

\subsection{Final algorithms}
\label{sec:final-algorithm} 
Let $S$ be the
input set of $n$ points. 
In the previous subsections we presented efficient sequential, MR and
streaming algorithms to construct $(1-\epsilon)$-coresets $T \subseteq S$ for all
variants of the DMMC problem considered in this paper. 
For all
settings and all variants, the final approximation algorithm can be obtained by
running a sequential
$\alpha$-approximation algorithm ${\mathcal{A}}$ on $T$, which will yield an
$(\alpha-\eta)$-approximate solution, where $\eta =
\epsilon/\alpha$. 
We remark that while ${\mathcal{A}}$ may exhibit very
high running time, the advantage of the coreset-based approach is that
the use of ${\mathcal{A}}$ is confined on a much smaller subset of the
input, retaining a comparable approximation quality while enabling the
solution of very large instances. 
In what follows, we concentrate on the partition and transversal
matroids, since for the general matroids no meaningful worst-case
time/space bounds can be claimed, even if we believe that our approach
can be of practical use even in the general case.  

\subsubsection{Sequential and streaming algorithms}
Theorem~\ref{thm:seqcoreset} shows that a $(1-\epsilon)$-coreset $T
\subseteq S$ can be computed in time $\BO{n|T|}$, where $|T| =
\BO{k(k/\epsilon)^D}$ for the partition matroid and
$|T|=\BO{k^2(k/\epsilon)^D}$) for the transversal matroid. For the
sum-DMMC variant, the local-search based, polynomial-time
$(1/2 - \gamma)$-approximation algorithm of \cite{AbbassiMT13}  can serve as
algorithm ${\mathcal{A}}$, yielding (for $\gamma=\epsilon$)
a final $(1/2 - 2\epsilon)$-approximation in polynomial time.  For all other variants, for which no
polynomial-time constant-approximation algorithms are known, we can
run an exhaustive search for the best solution on the coreset $T$,
yielding a $(1-\epsilon)$-approximation in time $\BO{n|T|+|T|^k}$. 
We observe that, in both cases, the dependence on the input size $n$ 
is merely linear and that for small values of $k$, which is
typical for many real-world applications, and for constant 
$\epsilon$ and $D$, the overall running time is within feasible 
bounds even for very large instances.

For the streaming setting, Theorem~\ref{thm:streamcoreset} states that
a $(1-\epsilon)$-coreset $T \subseteq S$ can be computed in one pass
with working memory $\BO{|T|}$, where the sizes of $T$ for the
partition and transversal matroids are the same as those claimed
before for the sequential setting.  Therefore, by running the
algorithm of \cite{AbbassiMT13} or an exhaustive search on $T$ at the
end of the pass, we obtain the same approximation guarantees stated
above.

\subsubsection{MapReduce algorithms}\label{sec:final-mapreduce}
Theorem~\ref{thm:MRcoreset} states that a $(1-\epsilon)$-coreset $T
\subseteq S$ can be computed in one MR round with linear total memory
and $\BO{n/\ell}$ local memory, where $\ell$ is the number of subsets
in the partition of $S$. $T$ has size $\BO{\ell k(k/\epsilon)^D}$
(resp., $\BO{\ell k^2(k/\epsilon)^D}$) for the partition (resp.,
transversal) matroid. By gathering $T$ in one reducer in a second
round, we may apply our novel sequential algorithms to extract the
final solution. Clearly, this second round requires local memory
$\BO{|T|}$. In order to balance the local-memory requirements between
the two rounds we can fix $\ell = \sqrt{n/k}$ (resp., $\ell =
\sqrt{n/k^2}$) for the partition (resp., transversal) matroid yielding
overall local-memory requirements of $\BO{\sqrt{nk} (k/\epsilon)^D}$
(resp., $\BO{\sqrt{n} k (k/\epsilon)^D}$).
(Observe that a better choice of $\ell$ yielding improved bounds on the
local memory could be made if $D$ were known.)  For what concerns the
quality of the solutions, it is easy to see that in this fashion we
can obtain the same spectrum of approximations as in the sequential
and streaming settings.


\section{Experiments}
\label{sec:experiments}

\begin{table}
  \caption{Datasets used in the experimental evaluation, $n$ is the number of elements.}
  \label{tab:datasets}
\begin{tabular}{lrrl}
\toprule
 & $n$ & Matroid $\rank$ & Matroid type \\
\midrule
Wikipedia & 5,886,692 & 100 & transversal \\
Songs & 237,698 & 89 & partition \\
\bottomrule
\end{tabular}
\end{table}

In this section, we report on three sets of experiments run on a
cluster of 16 machines, each equipped with a 18GB RAM and a 4-core
Intel I7 processor, connected by a 10Gbit Ethernet network.  
The first set (Subsection~\ref{sec:exp-sequential}) compares the
performance of our coreset-based approach with the state of the art in
the sequential setting.  The other two explore its applicability to very
large inputs, focusing on the Streaming
(Subsection~\ref{sec:exp-streaming}) and MapReduce models
(Subsection~\ref{sec:exp-mapreduce}), respectively.
The source code of our implementation
is publicly available\footnote{\url{https://github.com/Cecca/diversity-maximization}}.

As testbeds, we use two real-world datasets, whose characteristics are
summarized in Table~\ref{tab:datasets}.  One dataset is derived from a
recent dump of the English
Wikipedia\footnote{\url{https://dumps.wikimedia.org/backup-index.html},
  accessed on 2019-07-20}, comprising 5,886,692 pages.  Each Wikipedia page is
associated to a number of \emph{categories},  out of 1,102,435 overall
categories defined by the Wikipedia users, which naturally induce a
transversal matroid. Observe however that due to the sheer number of
categories, for any reasonable value of $k$, any subset of $k$ pages
would very likely be an independent set, thus making the matroid
constraint immaterial.  As a workaround, we applied the \emph{Latent
  Dirichlet Allocation} model~\cite{Blei2003} to derive a much smaller set of 100
\emph{topics}, which we use as new categories, together with a probability
distribution over these topics for each page. We then assign each
page to the most likely topics (probability $\ge 0.1$), thus
obtaining a transversal matroid of rank 100.  Finally, each page is mapped to a
25-dimensional real-valued vector using the \emph{Global Vectors for Word Representation}
(\texttt{GloVe}) model  \cite{pennington2014glove}.
The other dataset is a set of 237,698
songs~\footnote{\url{http://millionsongdataset.com/musixmatch/}}, each
represented by the bag of words of its lyrics and associated to a
unique \emph{genre}, out of a total of 16 genres. Since genres define
a partition of the dataset, they induce a partition matroid. For each
genre $g$, we fixed the associated cardinality threshold $k_g$ in the
matroid as the minimal nonzero value proportional to the
number of songs of the genre in the dataset, thus obtaining a
partition matroid of rank 89.  Each page is mapped to a sparse vector,
with a coordinate for each of the 5000 words of the dataset's
vocabulary, each set to the number of occurrences of the corresponding
word in the lyrics of the song.
(Along with the source code we also provide the scripts that we used
to preprocess the datasets.)

For both datasets we use as distance the metric version of the
\emph{cosine distance}~\cite{LeskovecRU14}.  All results reported are
obtained as averages over at least 10 runs.  To evaluate the stability
of the solution quality with respect to the initial data layout, the
dataset is randomly permuted before each run. All figures in the pdf
version of the paper can be clicked upon, so to open an online
interactive version which provides additional information about the
experiments pictured in the figure.

As discussed in Section\ \ref{sec:final-algorithm}, for the sum-DMMC
problem we can use the local search algorithm of\ \cite{AbbassiMT13}
to compute the final solution\footnote{%
For any coreset-based
  algorithm studied in this section, the final output is computed
  using the local search algorithm with $\gamma = 0$.}, whereas for
other diversity measures no approximation algorithm is known but we
can run an exhaustive search on a small enough coreset returned by our
strategy. For concreteness, we restrict the attention to the sum-DMMC
problem. In fact, the benefits of our coreset-based approach are
evident for the other problems, where the only alternative (with provable approximation guarantees) to our
strategy is an unfeasible
exhaustive search on the entire input.
 
\subsection{Sequential setting}
\label{sec:exp-sequential}

  \begin{figure*}
    \href{https://cecca.github.io/diversity-maximization/sequential.html}{
      \input{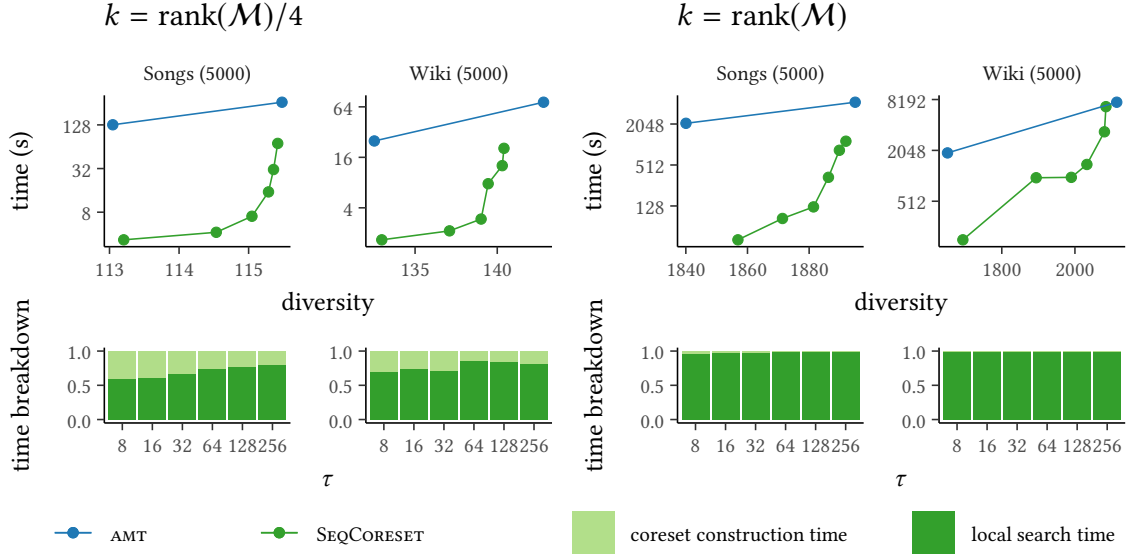}
    }
    \caption{Time vs. diversity for \textsc{amt} and \textsc{SeqCoreset} (top),
and running time breakdown for \textsc{SeqCoreset} (bottom).
In the top plots the y scale is logarithmic, and the $\tau$ parameter increases from left
to right for the \textsc{SeqCoreset} algorithm.
The best performance is towards the bottom-right corner: low running time and high diversity of the
solution.
}
    \label{fig:div-vs-time}
  \end{figure*}

We compare our sequential algorithm \algo{SeqCoreset},
described in Subsection~\ref{sec:seq-implementation},
against
the algorithm in \cite{AbbassiMT13}, 
which we refer to as \algo{amt} in the following.
We recall that \algo{amt} runs a local search over
the entire input, and features a parameter $\gamma$ to limit swaps to
those providing an improvement of a factor at least $(1+\gamma)$
to the current solution quality, thus exercising a tradeoff between
approximation guarantee ($\frac{1}{2}-\gamma$) and running time.
Albeit polynomial for constant $\gamma>0$, \algo{amt} is quite expensive,
since it may require to check a large number of candidate swaps
(possibly quadratic in the input size) where each check entails a call
to the independent set oracle, which can be a costly operation,
depending on the matroid type.  For this reason, in order to keep the
running times of \algo{amt} within reasonable limits,
we tested both algorithms on scaled-down versions of the datasets
obtained as samples of 5,000 elements drawn at random from each dataset.
However, we want to stress that \algo{SeqCoreset} is able to process
the entire datasets within reasonable time bounds, as will be shown by
the MapReduce experiments, where \algo{SeqCoreset} 
represents the case of parallelism 1 (see Subsection~ \ref{sec:exp-mapreduce}).
 
For what concerns \algo{SeqCoreset}, rather than using parameter
$\epsilon$, we control the radius of the clustering underlying the
coreset construction indirectly through the number of clusters $\tau$
to be found by \algo{gmm}, where larger values of $\tau$ yield smaller
radii, and thus correspond to smaller values of $\epsilon$. 
Specifically, we set $\tau$ to powers of two from $8$ to $256$.
For comparison, we ran several instances of \algo{amt} with values of
gamma in the range $[0, 0.9]$, by increments of 0.001.  To avoid
overcrowding the plot with too many points and to gauge the
performance of our algorithm versus its competitor in terms of time
and accuracy, we report the results of two specific runs of
\algo{amt}: the one with the value $\gamma$ returning the largest
diversity (ties broken in favor of fastest running time); and the one
with the value $\gamma$ returning a solution with a quality just below
the lowest one found by \algo{SeqCoreset}. All other tested runs of
\algo{amt} featured running times and diversities between those of
the two extreme runs. Also, for each dataset we used two
values of $k$, namely $k=\rank(\mathcal{M})$ and
$k=\rank(\mathcal{M})/4$, where $\mathcal{M}$ is the associated
matroid.

In the top row of Figure~\ref{fig:div-vs-time}, we plot the running
times of the two algorithms ($y$-axis) against  the diversity yielded by each
parameter configuration  ($x$-axis), so to compare the time taken by each
algorithm to compute a solution of similar quality.  The bottom row of
plots reports, for each parameter configuration of \algo{SeqCoreset},
the breakdown of the running time between the two components of the
algorithm: coreset construction (light green) and local search on the
coreset (dark green).  We observe that our algorithm returns solutions
of quality comparable to the ones returned by \algo{amt} and, as
emphasized by the logarithmic scale, in most cases it runs one or
two orders of magnitude faster. 

In all experiments, the coreset construction performed by
\algo{SeqCoreset} never dominates the overall running time.  For
$k=\rank(\mathcal{M})/4$, the coreset construction takes between 50\%
of the time (for $\tau = 8$) and 20\% of the total time (for
$\tau=256$).  For $k=\rank(\mathcal{M})$, building the coreset takes
negligible time compared to the total time (below 2\%).  This is a
consequence of both the limited input size and the expensive nature of
the local search task. In fact, in the context of the MapReduce experiments
reported in Subsection~\ref{sec:exp-mapreduce}, we also ran 
\algo{SeqCoreset} on both the full Songs and Wikipedia datasets, 
with $\tau=64$. The results, reported in Figure~\ref{fig:mr-experiments},
show that with a much larger dataset the coreset construction task dominates
the running time.

Finally, the shape of the curve in Figure~\ref{fig:div-vs-time}
representing the performance of our algorithm shows that the parameter
$\tau$ can be effectively used to control a tradeoff between accuracy
and running time, while parameter $\gamma$ of \algo{amt} seems less
effective in that respect.

\subsection{Streaming setting}
\label{sec:exp-streaming}

  \begin{figure}
    \href{https://cecca.github.io/diversity-maximization/streaming-h.html}{
      \input{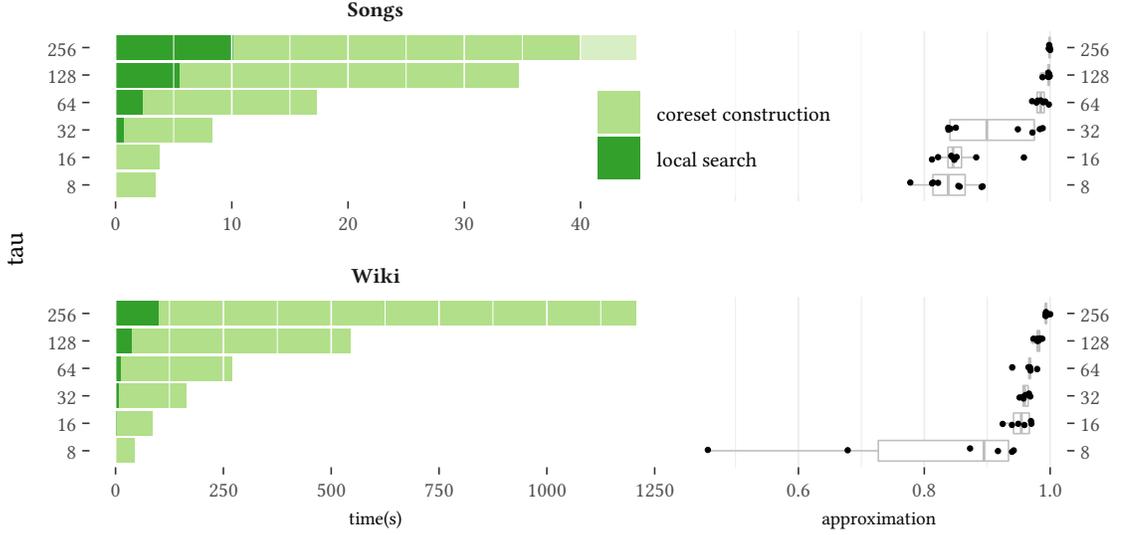}
    }
    \caption{Performance of the streaming algorithm.
For different coreset sizes, the box plots on the right show the diversity attained 
by different configurations of the algorithm. 
The bar charts on the left report the breakdown
of the overall running time.}
    \label{fig:streaming-experiments}
  \end{figure}

We evaluate the performance of our streaming strategy by analyzing the
relationship between coreset size, quality of approximation, and
running time. As done for the sequential setting, we implemented a
variant of {\sc StreamCoreset} (Algorithm~\ref{alg-streaming}) so to
control the number $\tau$ of clusters directly, rather than having the
value of $\tau$ to be determined implicitly as a function of the
approximation parameter $\epsilon$. Controlling the size of the
coreset directly also enables an easier comparison with the MapReduce
results presented in the next section.

To control $\tau$ directly, the implemented variant maintains in
variable $R$ an estimate of the \emph{radius} of the $\tau$-clustering
built so far, rather than an estimate of the diameter of the
dataset. For each point of the stream, if it falls within distance
$2\cdot R$ from any of the centers, it is handled using procedure
\textsc{Handle} of {\sc StreamCoreset}.  Otherwise it is
added as a new cluster center. As soon as there are more than $\tau$
clusters, the algorithm restructures the set of centers as in
{\sc StreamCoreset} and doubles $R$.  This variant is
reminiscent of the $k$-center streaming algorithm by
\cite{CharikarCFM04}, and, using an analysis similar to the one in
Subsection~\ref{sec:stream-implementation}, it can be shown that by
setting $\tau=\BT{(k/\epsilon)^D}$ it returns a
$(1-\epsilon)$-coreset.

We run the algorithm on the full Wikipedia and Songs datasets, fixing
$k=\rank(\mathcal{M})/4$.  As for the coreset size, we run the
algorithm so to build $\tau \in \{8,16,32,64,128,256\}$ clusters, each
containing the appropriate number of delegate points, depending on the
matroid type.

We report the results of these experiments in
Figure~\ref{fig:streaming-experiments}.  While the bars on the left
side of the picture plot the running time for each value of $\tau$,
the box-plots on the right report, for each dataset, the distribution
of the approximation ratios in the various runs.  Such ratios are
computed with respect to the best solution ever found by \emph{any run
  of any algorithm in any setting} (on the same dataset and for the
same value of $k$), with results close to 1 denoting better solutions.
We observe that despite the high dimensionality of the dataset,
increasing the coreset size has a beneficial effect on the solution
quality, at the expense of a roughly linear increase in the running
time. Note also that as the size of the coreset increases, all runs
tend to give solutions with diversity values which are more
concentrated.

We report the results obtained by \algo{StreamCoreset} with $\tau=64$
also in Figure~\ref{fig:mr-experiments} (red bars), to
compare with the performance of the algorithms in the other settings
for a fixed coreset size.
From the figure, we observe that, compared to \algo{SeqCoreset}, the streaming algorithm is around 7 
times faster on the Songs dataset,
and 4 times faster on Wikipedia.
This might be surprising, given that \algo{StreamCoreset} may perform
more distance computations than \algo{SeqCoreset} to build the coreset,
namely $\BT{\tau^2 n}$ instead of $\BT{\tau n}$ in the worst case.
However, this worst case rarely happens in practice. In fact,
\algo{StreamCoreset} is considerably more cache-efficient, since for
each point of the stream it computes the distance with $\BO{\tau}$
cluster centers, which easily fit into cache.  Conversely,
\algo{SeqCoreset} benefits less from data locality, since it iterates
$\tau$ times over the the entire input, not reusing information
already in the cache.
However, despite having a worse running time, \algo{SeqCoreset} yields solutions
of better quality than \algo{StreamCoreset}, as can be observed in the
box plots of Figure~\ref{fig:mr-experiments}.
This is not surprising: \algo{SeqCoreset} derives its coreset starting from the 
2-approximation \algo{gmm} clustering algorithm;
\algo{StreamCoreset}, instead, uses a clustering strategy similar 
to~\cite{CharikarCFM04}, which is an 8-approximation.


\subsection{MapReduce setting}
\label{sec:exp-mapreduce}

  \begin{figure*}
    \href{https://cecca.github.io/diversity-maximization/mapreduce-h.html}{
      \input{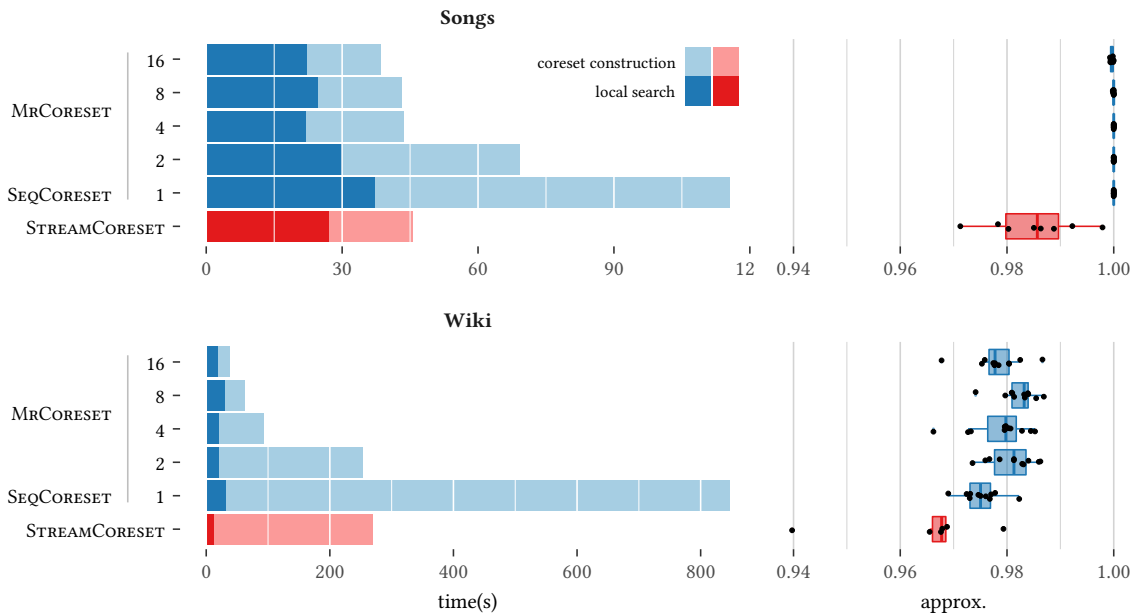}
    }
    \caption{Comparison between all the algorithms, with $\tau=64$,
on the full datasets.
The bars on the left report the running time.
For \textsc{MrCoreset} we report on the performance with 1, 2, 4, 8, and 16 machines.
The results with 1 machine correspond to the performance of \textsc{SeqCoreset}.
The portion of the bars
with saturated colors reports the time employed for the local search task,
the part with dimmed colors reports the coreset construction times.
Note that the scale is linear and that, for readability, each dataset has its own
time scale.
The box plots on the left report the quality of the solution 
found by each algorithm.
}
    \label{fig:mr-experiments}
  \end{figure*}

We implemented the MapReduce version of our coreset-based strategy
described in Section~\ref{sec:mr-implementation} (dubbed
\algo{MRCoreset}), using the Spark framework~\cite{spark}, and ran it
on the full Wikipedia and Songs datasets. As observed at the end of
Subsection~\ref{sec:exp-sequential} and confirmed by the streaming
experiments, with large inputs the main bottleneck of our approach
becomes the coreset construction task, where the whole input is
involved.  However, this task is fully parallelized in
\algo{MRCoreset}, thus yielding scalable performance. To emphasize
this aspect, we fixed a cluster granularity, namely $\tau=64$, which
provides a coreset whose size is small enough to limit the impact of
the local search task, but sufficient to embody a good quality
solution. Then, we ran \algo{MRCoreset}, on $\ell=1, 2, 4, 8$ and 16
machines, setting the number of clusters computed by each machine to
$\tau/\ell$, so to extract the final coreset always from a
$\tau$-clustering. For both datasets, we fixed
$k=\rank(\mathcal{M})/4$. 

By the choice of parameters, the final coreset is small enough that it
does not make sense to apply \algo{SeqCoreset} in the second round, as
described in Section \ref{sec:final-mapreduce}, since
\algo{SeqCoreset} would not reduce the size of the coreset
significantly, while it might obfuscate the evaluation of the impact
of parallelism on the solution quality.  Thus, we ran \algo{amt} with
$\gamma=0$ directly on the coreset computed in the first MapReduce
round, hence making the case $\ell=1$ coincide with
\algo{SeqCoreset} (in fact, in Subsection~\ref{sec:exp-sequential} we used
\algo{MRCoreset} with $\ell=1$ as implementation of
\algo{SeqCoreset}).

The results of our experiments are reported in the four plots of
Figure~\ref{fig:mr-experiments}. The left plots show the running times
broken down into coreset construction and local search time, under the
different levels of parallelism, while the right plots are box-plots
for the approximation ratios attained by \algo{MRCoreset}, computed as
described for the streaming setting. The figure reports also the
performance of \algo{SeqCoreset} and \algo{StreamCoreset} with $\tau = 64$,
so to allow a full comparison of all our algorithms.
Indeed, the bars of these two algorithms show clearly 
that, on large inputs, the majority of the work goes
into the coreset construction, and the bars corresponding to larger
levels of parallelism show that such construction scales well.  This
scalability effect is more evident on the Wikipedia dataset which, due
to its larger size, is able to take better advantage of the available
parallelism.  Not surprisingly, on this larger dataset the
coreset construction scales more than linearly. 
This is due to the fact that the complexity of the clustering 
required to compute the local coresets is roughly inversely proportional
to $\ell^2$, since $\tau/\ell$ clusters must be computed in
subsets of the dataset $S$ of size  $|S|/\ell$.
Importantly, we have that parallelizing the coreset construction,
which may in theory yield coresets of worse quality, does not seem to
affect the quality of the final solution significantly.

As for a comparison with the streaming algorithm, consider the
red bars of the figure.  For the Songs dataset, which is not very
large, the performance of the streaming algorithm (which employs a
single processor) is comparable with the performance of the MapReduce
algorithm with 16 processor.  This is due to the overhead of
communication and synchronization between machines inherent in the
Spark platform. On the Wikipedia dataset, on the other hand,
the benefits of parallelization emerge more evidently, and the
running time of the streaming algorithm is already matched with parallelism
2 but with a better solution quality.

Finally, we remark that the high complexity of \algo{amt} makes it
impractical for such large inputs, thus ruling out a direct comparison
with \algo{MRCoreset}.  However, by pairing the results of the
comparison between \algo{amt} and \algo{SeqCoreset} from
subsection~\ref{sec:exp-sequential} and between \algo{SeqCoreset} and
\algo{MRCoreset} from this subsection, we can infer that the latter
holds the promise to provide solutions of similar quality as
\algo{amt} but in a scalable fashion.


\section{Conclusions}
\label{sec:conclusions}

Coreset-based strategies provide an effective way of 
processing massive datasets
by building a succint summary of the input dataset $S$, which can then be analyzed
with a (possibly computationally-intensive) sequential algorithm of
choice. In this paper, we have seen how to leverage coresets to build
fast sequential, MapReduce and Streaming algorithms for diversity
maximization under matroid constraints  for a wide family of DMMC variants. 
For the sum-DMMC variant, our algorithms feature an accuracy which can be made
arbitrarily close to that of the  state-of-the art (computationally expensive)
sequential algorithms, while for all other variants they provide the 
the first viable $(1-\epsilon)$-approximate solutions in all of the aforementioned computational fameworks,
by confining exhaustive search  to the coreset, whose size is independent of $|S|$.
For the important cases of partition and transversal matroids, and
under resonable assumptions on the dimensionality of the dataset, the algorithms require work linear in the
input size, and, as demonstrated by experiments conducted on real world
datasets, their performance can be orders of magnitude faster than that
of existing algorithms.  Moreover, the Streaming and MapReduce versions of the
algorithms can be effectively employed in big data scenarios where
solutions of very large instances are sought.

In our algorithms the coreset size, which cannot grow too large to
avoid incurring large overheads in the extraction of the final
solution, exhibits an exponential dependency on the doubling dimension
$D$ of the input dataset. A challenging, yet important open problem is
to provide a tighter analysis of our algorithms, if at all possible,
or to develop improved strategies that retain efficiency even for
large values of $D$.

Recall that for the transversal matroid we made the assumption that
each point belongs to a constant number of categories. In fact, our
sequential and MapReduce algorithms can be easily modified to work
without this assumption, while maintaining the same quality.  Another
interesting open problem is to modify our streaming algorithm so that
the assumption can be lifted.

Finally, we wish to remark that our coreset-based approach for diversity maximization
under matroid constraints does not 
yield efficient algorithms for another well-studied variant (dubbed \emph{min-DMMC}), where the
diversity function of a subset of points is defined as the minimum distance between these
points. Devising solutions for  this important variant 
remains an interesting open problem.

\bibliographystyle{ACM-Reference-Format}
\bibliography{references}

\end{document}